\documentclass[useAMS, usenatbib]{biom}
\onecolumn
\usepackage{xcolor}
\definecolor{darkpastelgreen}{rgb}{0.01, 0.75, 0.24}

\usepackage{natbib}
\usepackage{mathtools}
\usepackage{mathrsfs}
\usepackage{setspace}
\usepackage{dsfont}
\usepackage{caption}
\usepackage{booktabs}
\usepackage{bm}
\usepackage{tikz}
\usetikzlibrary{fit,shapes,shadows,arrows,positioning,graphs}
\usepackage{soul}
\usepackage{nicefrac}
\usepackage{accents}
\usepackage{algpseudocode}
\usepackage{algorithm}
\usepackage[colorlinks = true,
            linkcolor = black,
            urlcolor  = black,
            citecolor = black]{hyperref}
\algnewcommand\Input{\item[{\textbf{\quad Input:}}]}
\algnewcommand\Output{\item[{\textbf{\quad Output:}}]}
\algnewcommand\Data{\item[{\textbf{\quad Data:}}]}

\usepackage{multirow}
\usepackage{array}
\usepackage{subfig}
\usepackage{graphicx}

\usepackage{verbatim}
\usepackage{soul}
\usepackage[english]{babel}
\usepackage[title]{appendix}

\DeclareMathOperator{\argmin}{argmin }
\DeclareMathOperator{\Tr}{Tr}

\newcommand{\dd}{\mathrm{d}}
\newcommand{\ind}[1]{\mathds{1}_{#1}}
\newcommand{\norm}[1]{\|#1\|}
\newcommand{\cY}{\mathcal Y}
\newcommand{\cD}{\mathcal D}
\newcommand{\cC}{\mathcal C}
\newcommand{\cG}{\mathcal G}
\newcommand{\cN}{\mathcal N}
\newcommand{\cQ}{\mathcal Q}
\newcommand{\cU}{\mathcal U}
\newcommand{\cR}{\mathcal R}
\newcommand{\cH}{\mathcal H}
\newcommand{\cB}{\mathcal B}
\newcommand{\cS}{\mathcal S}

\newcommand{\cL}{\mathcal L}
\newcommand{\cF}{\mathcal F}
\newcommand{\R}{\mathds R}
\newcommand{\N}{\mathds N}
\newcommand{\E}{\mathds E}
\renewcommand{\P}{\mathds P}
\newcommand{\bSigma}{\textbf{$\Sigma$}}

\title[FLASH]{An efficient joint model for high dimensional longitudinal and survival data\\ via generic association features}

\author{Van Tuan Nguyen$^{1,2,*}\email{vantuan.nguyen@califrais.fr}$, Adeline Fermanian$^{2}$, Antoine Barbieri$^{3}$,\\
\bf{Sarah Zohar$^{4,5}$, Anne-Sophie Jannot$^{4,5,6}$,  Simon Bussy$^{1}$, and Agathe Guilloux$^{4,5}$}\\
$^{1}$LOPF, Califrais’ Machine Learning Lab, Paris, France, $^{2}$Laboratoire de Probabilités, Statistique et \\Modélisation, LPSM, Univ. Paris Cité, F-75005, Paris, France, $^{3}$Univ. Bordeaux, INSERM, U1219, France,\\ $^{4}$INSERM, Centre de recherche des Cordeliers, Univ. Sorbonne, Univ. Paris Cité, F-75006, Paris, France, \\$^{5}$HeKA, Inria Paris, F-75015, Paris, France, $^{6}$French National Rare Disease Registry (BNDMR), \\Greater Paris University Hospitals (AP-HP), Paris, France.}

\begin{document}





\pagerange{\pageref{firstpage}--\pageref{lastpage}} 




\label{firstpage}


\begin{abstract}
This paper introduces a prognostic method called FLASH that addresses the problem of joint modelling of longitudinal data and censored durations when a large number of both longitudinal and time-independent features are available.  In the literature, standard joint models are either of the shared random effect or joint latent class type. Combining ideas from both worlds and using appropriate regularisation techniques, we define a new model with the ability to automatically identify significant prognostic longitudinal features in a high-dimensional context, which is of increasing importance in many areas such as personalised medicine or churn prediction. We develop an estimation methodology based on the EM algorithm and provide an efficient implementation. The statistical performance of the method is demonstrated both in extensive Monte Carlo simulation studies and on publicly available medical datasets.
Our method significantly outperforms the state-of-the-art joint models in terms of C-index in a so-called ``real-time'' prediction setting, with a computational speed that is orders of magnitude faster than competing methods. In addition, our model automatically identifies significant features that are relevant from a practical point of view, making it interpretable, which is of the greatest importance for a prognostic algorithm in healthcare.
\end{abstract}

%



\maketitle


%

\section{Introduction}
\label{introduction}

In healthcare, it is increasingly common to record the values of longitudinal features (e.g. biomarkers such as heart rate or haemoglobin level) up to the occurrence of an event of interest for a subject, such as rehospitalisation, relapse, or disease progression. Moreover, in large observational databases such as claims databases, with electronic health records, the amount of recorded data per patient is often very large and growing over time.
However, there is currently no tool that can simultaneously process high-dimensional longitudinal signals and perform real-time predictions, i.e. give predictions at any time after only one estimation/training step. While landmark approaches~\citep[see, e.g.,][]{devaux2022individual} can handle high-dimensional longitudinal features, they require a separate training for each prediction time. An alternative is to use ``joint modeling'' to handle longitudinal and survival outcomes together.

The latter has received considerable attention in the last two decades~\citep{tsiatis2004joint, rizopoulos2011bayesian, hickey2016joint}. More specifically, it consists of defining $(i)$ a time-to-event model, $(ii)$ a longitudinal marker model, and $(iii)$ linking both models via a common latent structure. Numerical studies suggest that these approaches are among the most satisfactory for incorporating all longitudinal information into a survival model~\citep{yu2004joint}, and are better than landmark approaches, which use only information from individuals at risk at the landmark time~\citep[see, e.g.,][]{devaux2022individual}. 
They have the additional advantage of making more efficient use of the data as information on survival is also used to model the longitudinal markers. More importantly, they require only one training, regardless of the number of prediction times.

There are two main approaches to linking longitudinal and survival models. On the one hand, in shared parameter joint models (SREMs), characteristics of the longitudinal marker, typically some random effects learned in a linear mixed model, are included as covariates in the survival model~\citep{wulfsohn1997joint, andrinopoulou2016bayesian}.
On the other hand, joint latent class models (JLCMs), inspired by mixture-of-experts modelling~\citep{masoudnia2014mixture}, assume that the dependence between the time-to-event and the longitudinal marker is fully captured by a latent class structure \citep{lin2002latent, proust2014joint}, which amounts to assuming that the population is heterogeneous and that there are homogeneous latent classes that share the same marker trajectories and prognosis. JLCMs offer a computationally attractive alternative to SREMs, especially in a high-dimensional context. These two models are illustrated in Figure \ref{fig:graphical_srem_jlcm}.

\begin{figure}[!htb]
\centering
\resizebox{.45\textwidth}{!}{
\begin{tikzpicture}
\tikzstyle{main}=[circle, minimum size = 18mm, ultra thick, draw =black!80, node distance = 13mm,text centered]
\tikzstyle{connect}=[-latex, thick]
\tikzstyle{box}=[rectangle, draw=black!30]
\begin{scope}
  \node[main, fill = white!30] (T) [] {$\bm{T}$};
  \node[main, fill = white!30] (Y) [below left=of T] {$\bm{Y}$};
  \node[main, fill = white!30] (X) [above left=of T] {$\bm{X}$};
  
  \path [->,draw,ultra thick]
        (X) edge (T)
        (Y) edge (T);
\end{scope}

\begin{scope}[xshift=3cm]
  \node[main, fill = white!30] (X) [] {$\bm{X}$};
  \node[main, fill = black!30] (G) [right=of X] {$\bm{G}$};
  \node[main, fill = white!30] (Y) [above right=of G] {$\bm{Y}$};
  \node[main, fill = white!30] (T) [below right=of G] {$\bm{T}$};
  
  \path [->,draw,ultra thick]
        (X) edge (G)
        (G) edge (T)
        (G) edge (Y);
\end{scope}   
\end{tikzpicture}
}
\caption{Graphical representation of SREMs (left) and JLCMs (right). The variable $X$ represents time-independent features, $Y$ the longitudinal markers, $T$ the time-to-event, and $G$ the latent class membership.}
\label{fig:graphical_srem_jlcm}
\end{figure}
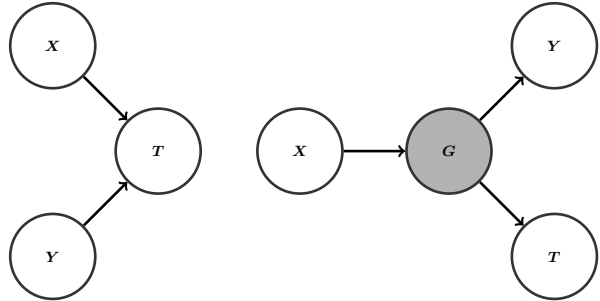

Unfortunately, joint models have predominantly focused on univariate longitudinal markers \citep{andrinopoulou2020integrating}. To adapt such models to a multivariate setting, the common approach is to fit multiple univariate joint models separately to each longitudinal marker~\citep{wang2012joint}, which does not account for interactions between longitudinal markers~\citep{jaffa2014joint, kang2022consistent,lin2002maximum}. Furthermore, issues arising from the high-dimensional context---e.g. computational power, limits of numerical estimation---have, to our knowledge, never been considered in the analyses, and the number of longitudinal markers considered in numerical studies remains very low, typically up to $5$ \citep{hickey2016joint,murray2022fast,rustand2024fast}. 

The aim of this article is to propose a new joint model called FLASH (Fast joint model for Longitudinal And
Survival data in High dimension), together with an efficient inference methodology. The model is inspired by both JLCMs and SREMs, but is designed to scale to high-dimensional longitudinal markers. The general idea is to use generic features extracted from the longitudinal markers directly in the survival model and to use regularization in an Expectation-Maximization (EM) algorithm. The main difference with SREMs is that these features, called association features, are assumed to be independent of the modeling assumptions in the longitudinal submodel.
As a result, the model is very efficient to train---the likelihood is closed-form in a Gaussian setting and does not require Monte Carlo approximations, as is often the case with SREMs, making it suitable for high-dimensional longitudinal markers. Moreover, it allows the use of high-dimensional and generic feature extraction functions that characterize longitudinal markers, rather than just the random effects, resulting in a model that is generic enough to be adapted to different types of longitudinal markers and prior information on the problem. The use of elastic net and sparse group lasso regularization enables automatic selection of relevant association features, resulting in an interpretable model that retains only the significant longitudinal markers

Finally, our model allows us to define a ``real-time'' prediction methodology where, once the parameters of the model have been learned, we can compute a predictive marker that, given only the time-independent and longitudinal features up to a given point in time, outputs a risk for each subject at that point in time. It differs from traditional approaches \citep{proust2014joint} that require knowledge of the survival labels, which are unknown in the ``real-time'' prediction setting.

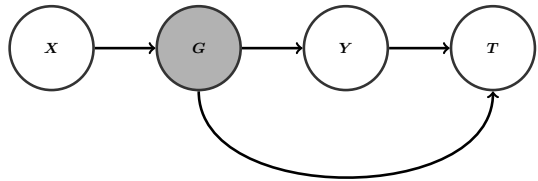
\begin{figure}[!htb]
 \centering
\resizebox{.4\textwidth}{!}{
\begin{tikzpicture}
\tikzstyle{main}=[circle, minimum size = 18mm, ultra thick, draw =black!80, node distance = 13mm,text centered]
\tikzstyle{connect}=[-latex, thick]
\tikzstyle{box}=[rectangle, draw=black!30]
  \node[main, fill = white!30] (X) [] {$\bm{X}$};
  \node[main, fill = black!30] (G) [right=of X] {$\bm{G}$};
  \node[main, fill = white!30] (Y) [right=of G] {$\bm{Y}$};
  \node[main, fill = white!30] (T) [right=of Y] {$\bm{T}$};
  
  \path [->,draw,ultra thick]
        (X) edge (G)
        (G) edge (Y)
        (Y) edge (T)
        (G) edge[bend right=90] (T);
\end{tikzpicture}}
\caption{Graphical representation of the FLASH model. The variable $X$ represents time-independent features, $Y$ the longitudinal markers, $T$ the time-to-event and $G$ the latent class membership of an individual.}
\label{fig:graphical_flash}
\end{figure}

In summary, we introduce a new method for predicting survival risk with high-dimensional longitudinal features that is both interpretable and computationally efficient, thus providing a powerful tool for real-time clinical decision making, for example in patient monitoring.

A precise description of the model is given in Section~\ref{sec:method}. In
Section~\ref{sec:inference}, we present our assumptions and inference methodology based on maximizing the likelikood with a regularized variant of the EM algorithm. Section~\ref{sec:evaluation_methodology} introduces our evaluation methodology and the competing methods. In Section~\ref{sec:experiment}, we apply our method to datasets from two simulation studies and three publicly available medical datasets (PBCseq, AIDS, Sepsis). We show that FLASH outperforms the competitors and selects relevant features from a medical perspective. Finally, we discuss the obtained results in Section~\ref{sec:conclusion}. The code to implement the model and reproduce the experiments is publicly available at \url{https://github.com/Califrais/flash.git}.
\section{Model}
\label{sec:method}

In this section we describe the FLASH model, which consists of three sub-models: a multinomial logistic regression defining the probability of belonging to a latent class, a generalized linear mixed model for each latent class describing the evolution of the longitudinal markers, and finally a Cox class-specific survival model. In all the following, we consider a set of $n$ independent and identically distributed (i.i.d.) subjects. For each subject $i \in \{1, \dots, n \}$ we are given some longitudinal markers $Y_i$, time-independent features $X_i \in \R^p$, a right-censored time-to-event $T_i \in \R^+$, and a censoring variable $\Delta_i \in \{0, 1\}$. 

\subsection{Latent class membership } \label{subsec:latent_class}
We assume that the population consists of $K \in \N^\ast$ latent groups. To each subject $i \in \{1, \ldots, n\}$, we associate a categorical latent variable $g_i \in \{1, \dots, K\}$, which encodes its latent class membership. Then, denoting by $X_i \in \R^p$ the $p$-dimensional vector of time-independent features, the latent class membership probability is assumed to take the multinomial logistic regression form~\citep{bohning1992multinomial}, for any $k \in \{1, \dots, K\}$,
\begin{equation}
  \label{eq:pi}
  \P(g_i=k) = \dfrac{e^{X_i^\top\xi_k}}{\sum_{j=1}^{K}e^{X_i^\top\xi_j}},
\end{equation}
where $\xi_k \in \R^p$ denotes a vector of coefficients for class $k$. This submodel is similar to JLCMs or the C-mix model~\citep{bussy2019c} and assumes that latent-class membership depends only on time-independent features, with the vector $\xi_k$ quantifying the effect of each time-independent feature in $X_i$ on the probability that subject $i$ belongs to the $k$-th latent class.
The optimal number of latent classes $K$ can be selected with the Bayesian information criterion (BIC)~\citep{hastie2009elements}, see Section~\ref{sec:K_sel} of the Supplementary Materials for more details. Note that modeling the probability of latent class membership using multinomial logistic regression is a common assumption in joint models~\citep{lin2002latent, proust2014joint}. However, other models such as hidden Markov models can also be used to represent this probability~\citep{bartolucci2015discrete, bartolucci2019shared}. We assess the sensitivity to this choice by some experiments in Section~\ref{sec:sensity} of the Supplementary Materials.


\subsection{Class-specific longitudinal model} \label{subsec:longitudinal}

For each subject $i \in \{1, \ldots, n\}$, we are given $L \in \N^\ast$ longitudinal markers. We let, for any $\ell \in  \{1, \ldots, L \}$, $Y_i^\ell=\Big( y_i^\ell\big(t_{i1}^\ell\big), \ldots, y_i^\ell \big(t_{i n_i^\ell}^\ell\big) \Big)^\top \in \R^{n_i^\ell}$ be the vector of repeated measures of a theoretical longitudinal marker $y_i^\ell$ at observation times (or follow-up visits) $0 \leq t_{i1}^\ell <  \cdots < t_{in_i^\ell}^\ell$. Note that the observation times $t_{ij}^\ell$, $j=1, \dots, n_i^\ell$, can differ between subjects as well as between longitudinal markers, which makes the assumptions on the sampling mechanism very weak. In particular, this setting encapsulates many scenarios considered as missing data, where one individual is not measured while another one is, or where the longitudinal maker of one individual is missing, which here both simply correspond to removing some timestamps from the corresponding grids.

We assume a class-specific generalized linear mixed model (GLMM) for each longitudinal marker, which is a classical model for longitudinal data~\citep{fitzmaurice2012applied,hickey2016joint}. The GLMM is chosen according to the nature of the markers: Gaussian linear model for continuous markers, logistic regression for a binary marker, Poisson regression for counts. For the continuous markers, given a latent class $g_i=k$, for the $\ell$-th marker at time $t \in \{t_{i1}^\ell, \dots, t_{in_i^\ell}^\ell\}$, we then have
\begin{equation} \label{eq:gaussian_model_Y_i_ell}
    y_i^{\ell}(t^{\ell}_{ij}) \, | \,  b_i^\ell, \, g_i=k \sim \mathcal{N}(m_{ik}^{\ell}(t^{\ell}_{ij}), \phi_\ell),
\end{equation}
where the variance $\phi_\ell \in \R^+$ is an estimated parameter and the mean $m_{ik}^{\ell}$ is defined by
\begin{equation*}
  m_{i k}^\ell(t) = u^\ell(t)^\top\beta_k^\ell + v^\ell(t)^\top b_i^\ell, 
 \end{equation*}
 where $u^\ell(t) \in \R^{q_\ell}$ is a row vector of time-varying features with corresponding unknown fixed effect parameters $\beta_k^\ell\in \R^{q_\ell} $, and $v^\ell(t) \in \R^{r_\ell}$ is a row vector of time-varying features with corresponding random effect $b_i^\ell \in \R^{r_\ell}$.  
 Flexible representations for $u^\ell(t)$ can be considered using a vector of time monomials $u^\ell(t) = (1, t, t^2, \ldots, t^\alpha)^\top,$ with $\alpha \in \N^+$. We use $\alpha=1$ in all our experiments but higher orders could be used. We also let $v^\ell(t) = (1, t)^\top$.

Classically, the random effects component is assumed to follow a zero-mean multivariate normal distribution, that is, $b_i^\ell \sim \cN(0, D^{\ell \ell})$ with $D^{\ell \ell} \in \R^{r_\ell \times r_\ell}$ the variance-covariance matrix. To account for the dependence between the different longitudinal markers, we let $\text{Cov}[b_i^\ell,b_i^{\ell'}] = D^{\ell \ell'}$ for $\ell \ne \ell'$, where $\text{Cov}[\cdot, \cdot]$ denotes the covariance matrix of two random vectors, and we denote by $D = (D^{\ell \ell'})_{1 \leq \ell, \ell' \leq L}$
the global variance-covariance matrix which is common to all latent classes. 
Note that this variance–covariance matrix $D$ can be easily extended to be class-specific. We assume that all dependencies between longitudinal markers are encapsulated in this matrix $D$, which is summarized in the following assumption.
\begin{assumption}
\label{indep-hyp-1}
 For any $\ell \in \{1, \ldots, L\}$ and any $i \in \{1, \dots, n\}$, the longitudinal markers $Y_i^\ell$ are pairwise independent conditionally on $b^\ell_i$ and $g_i$.
\end{assumption} 
This is a standard modelling assumption in joint models~\citep[see, e.g.,][]{tsiatis2004joint}. Then, if we concatenate all longitudinal measurements and random effects of subject $i$ in, respectively, $Y_i = \big(Y_i^{1\top} \cdots \, Y_i^{L\top} \big)^\top \in \R^{n_i} \quad \text{and} \quad b_i = \big({b_i^1}^\top \cdots \, b_i^{L \top} \big)^\top \in \R^r$
with $n_i = \sum_{\ell=1}^L n_i^\ell$ and $r = \sum_{\ell=1}^L r_\ell$, a consequence of Assumption~\ref{indep-hyp-1} and Equation \eqref{eq:gaussian_model_Y_i_ell} is that
\begin{equation}
    \label{eq:Y_i_density}
  Y_i \, | \, b_i, g_i = k \sim \cN\big(M_{ik},\Sigma_i\big),
\end{equation}
where $M_{ik} = \big(m_{ik}^{1}(t^{1}_{i1}), \ldots, m_{ik}^{1}(t^{1}_{in_i^1}) ,\ldots, m_{ik}^{L}(t^{L}_{i1}), \ldots, m_{ik}^{L}(t^{L}_{in_i^L})\big)^\top \in \R^{n_i}$ and $\Sigma_i$ is the diagonal matrix whose diagonal is $(\phi_1{\textbf{1}_{n_i^1}}^\top, \ldots, \phi_L {\textbf{1}_{n_i^L}}^\top)^\top \in \R^{n_i}$ where $\mathbf{1}_m$ denotes the vector of $\R^m$ having all coordinates equal to one.

To extend the GLMM to other types of longitudinal markers (for example, binary or count data), the key point is that the distributions of $b_i$ and of $Y_i \, | \, b_i, g_i = k$ should be conjugate, so that the likelihood is tractable. It is possible to extend the model to non-conjugate distributions via numerical integration methods but this is not immediate. A detailed discussion of this is given in Section~\ref{sec:E-step} of the Supplementary Materials. For simplicity, we restrict ourselves to the Gaussian case in this article.

\subsection{Class-specific Cox survival model} \label{subsec:survival}

We place ourselves in a classical survival analysis framework. Let the non-negative random variables $T_i^\star$ and $C_i$ be the time to the event of interest and the censoring time, respectively. We then define the observed time $T_i = T_i^\star \wedge C_i$ and censoring indicator $\Delta_i = \ind{\{T_i^\star \leq C_i\}},$ where $a \wedge b$ denotes the minimum between two real numbers $a$ and $b$, and $\ind{\{\cdot\}}$ is the indicator function which takes the value $1$ if the condition in $\{\cdot\}$ is satisfied, and $0$ otherwise. We denote by $\cY_i^\ell(t^{-}) = \big(y_i^\ell(t_{i1}^\ell), \ldots, y_i^\ell(t_{iu}^\ell)\big)_{0 \leq t_{iu}^\ell < t}$
the subset of $Y_i^\ell$ formed from observations up to time $t$ and by $\cY_i(t^{-})$ the concatenation of the history of all observed longitudinal markers up to $t$.
Then we consider $M \in \N^+$ user-defined feature extraction functions $\Psi_m: \cY_i^\ell(t^{-}) \to \Psi_m(\cY_i^\ell(t^{-})) \in \R$, $m \in \{1, \dots, M \}$, which characterise the longitudinal markers. The set of features $\big(\Psi_m(\cY_i^\ell(t^{-}))\big)_{1 \leq m \leq M}$ should be rich enough to capture all dependencies between longitudinal markers and time-to-event, and is discussed in more detail below. To quantify the effect of the longitudinal markers on time-to-event, we then use an extension of the Cox relative risk model~\citep{Cox1972JRSS}, which allows time-varying covariates and was firstly introduced in~\citet{andersen1982cox}. Note that this model does not fulfill the classical Cox model's assumption of a constant proportional hazard over time. The hazard function in this model takes the form
\begin{equation}
  \label{eq:intensity-model}
  \lambda(t \, | \, \cY_i(t^{-}), g_i = k) = \lambda_0(t) \exp \Big(\sum_{\ell=1}^L \sum_{m=1}^M \Psi_m \big(\cY_i^\ell(t^{-}) \big){\gamma_{k,m}^\ell} \Big) = \lambda_0(t) \exp \big(\psi_i(t)^\top \gamma_k\big),
\end{equation}
where $\lambda_0$ is an unspecified baseline hazard function that does not depend on $k$, $\gamma_{k,m}^\ell \in \R$ the joint association parameters, which are the only class-specific objects in this model. We concatenate them in $\gamma_k = (\gamma_{k,1}^1, \ldots, \gamma_{k,M}^1, \ldots, \gamma_{k,1}^L, \ldots, \gamma_{k,M}^L )^\top \in \R^{LM}$ and define $\psi_i(t) = \Big(\Psi_{1}\big(\cY_i^1(t^{-}) \big), \ldots, \Psi_M\big(\cY_i^1(t^{-}) \big), \ldots, \Psi_{1}\big(\cY_i^L(t^{-}) \big), \ldots, \Psi_M\big(\cY_i^L(t^{-}) \big) \Big)^\top \in \R^{LM}$.

This model can be viewed as a generalization of SREMs~\citep{lin2002maximum, rizopoulos2011bayesian}, which have hazard functions of the form $\lambda_0(t) \exp \big(\sum_{\ell=1}^{L} \phi(b^\ell_i, t)^\top \gamma^\ell \big)$, where the association between the longitudinal and survival models is captured by the random effects $b^\ell_i$. The key idea of our model lies in the following assumption.
\begin{assumption}
\label{indep-hyp-2}
    For any time $t\geq0$, the hazard rate at time $t$ conditionally on the history of the longitudinal markers up to $t^{-}$ and $g_i$ is independent of $b_i$.
\end{assumption}
Any feature extraction function $\Psi_m$ of the longitudinal markers can be considered in the hazard function. Simple examples of such functions are the maximum or the sum (or the mean) of the longitudinal features, respectively defined by
\begin{align*}
    \Psi_m(\cY_i^\ell(t^{-})) &= \max_{j : t_{ij}^\ell < t}\big\{y_i^\ell(t_{i1}^\ell), \ldots, y_i^\ell(t_{ij}^\ell) \big\}, \quad \Psi_m(\cY_i^\ell(t^{-})) = \sum_{j :t_{ij}^\ell < t } y_i^\ell(t_{ij}^\ell).
\end{align*}
For illustration, Figure~\ref{fig:example_asso} shows how these association features change with time.
\begin{figure}
    \centering
    \includegraphics[scale=.2]{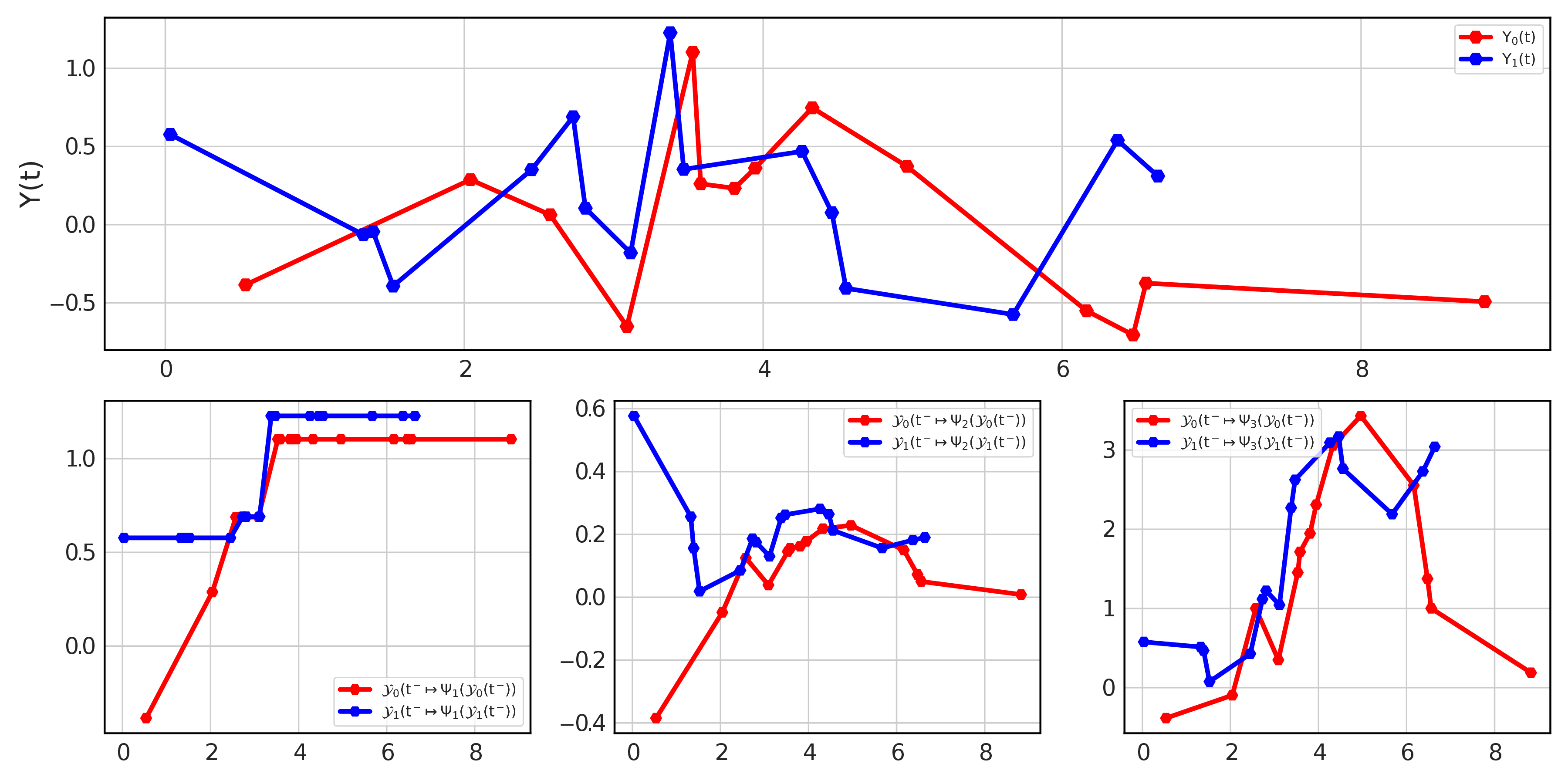}
    \caption{On the \textbf{top}, observed longitudinal markers of 2 individuals from the FLASH\_simu dataset. On the \textbf{bottom}, association features applied on the observed longitudinal markers: the maximum (\textbf{left}), the mean (\textbf{center}), and the sum (\textbf{right}).}
    \label{fig:example_asso}
\end{figure}
In practice, our rationale is to use a variety of feature extraction functions $\Psi_m$, such as absolute energy over time, statistics on autocorrelation, or Fourier and wavelet basis projections, and then perform feature selection via regularisation to learn which ones are predictive for the underlying task. This will be described in more detail in Section~\ref{sec:penalized-obj}.  Note that a crucial aspect of this model is that the extracted vector, also called extracted features or association features, $\psi_i(t)$, does not depend on the modelling assumptions in the longitudinal submodel of Subsection \ref{subsec:longitudinal} --- that is, does not depend on $b_i$ other than through the history $\cY_i(t^{-})$.

The FLASH model is summarised in Figure \ref{fig:graphical_flash} which shows that our model is a combination of SREMs and JLCMs where both random effects and latent classes account for the dependence between longitudinal markers and time-to-event.

\section{Inference}
\label{sec:inference}

Now that we have introduced all the components of our model, in this section we derive the form of its likelihood, present the regularisation strategy that deals with the high dimensionality of the data, and finally present our variant of the EM algorithm used to minimise the penalised negative log-likelihood.


\subsection{Likelihood}

Consider a training cohort of $n$ i.i.d. subjects 
$\cD_n = \big( (X_1, Y_1, T_1, \Delta_1), \ldots, (X_n, Y_n, T_n, \Delta_n) \big).$
For simplicity, we slightly abuse notation and use the same notation $f^\star$ for the true (joint) density or probability mass function of the various random variables in our model. Similarly, we denote by $f_\theta$ the candidates for estimating the densities $f^\star$ that satisfy the model assumptions of Section~\ref{sec:method}, where we have concatenated in $\theta$ all $P \in \N^+$ unknown parameters:
\[ \theta = \big(\xi_1^\top, \ldots, \xi_{K}^\top, \beta_1^\top, \ldots, \beta_{K}^\top, \phi^\top, D, \lambda_0(\tau_1), \ldots, \lambda_0(\tau_J), \gamma_1^\top, \ldots \gamma_{K}^\top\big)^\top \in \R^P ,\]
where $\beta_k = ({\beta_k^1}^\top \dots {\beta_k^\ell}^\top)^\top \in \R^q$ with $q = \sum_{\ell=1}^L q_\ell$ for any $k \in \{1, \dots, K\}$ , $\phi = (\phi_1, \ldots, \phi_L)^\top$ and where we use the vectorization of the matrix $D$ although this is not written explicitly. Note that we classically \citep[see, e.g.,][]{klein1992semiparametric} estimate $\lambda_0$ by a function taking mass at each failure time $\tau_j \in (\tau_1, \ldots, \tau_J)$, where $(\tau_1, \ldots, \tau_J)$ denote the $J \in \N^+$ unique failure times (obtained from $(T_1, \dots, T_n)$ removing the duplicates and keeping only the uncensored times $T_i$ for which $\Delta_i=1$). In this way, the estimation of the function $\lambda_0$ amounts to the estimation of the vector $\big(\lambda_0(\tau_1), \dots, \lambda_0(\tau_J)\big)$.

First, conditioning on the latent classes, we have
\begin{align*}
      f^\star(T_i, \Delta_i, Y_i) = \sum_{k=1}^K f^\star(g_i=k) f^\star(T_i, \Delta_i |Y_i, g_i = k) f^\star(Y_i |g_i = k).
\end{align*}
This yields the negative log-likelihood
\begin{align}
\label{eq:log-lik}
 \cL_n(\theta) = -n^{-1} \sum_{i=1}^n \log \sum_{k=1}^{K} f_\theta(g_i = k) f_\theta(T_i, \Delta_i \, | \, Y_i, g_i = k) f_\theta(Y_i \, | \, g_i = k).
\end{align}
Assuming that both the censoring mechanism and the stochastic mechanism generating the observation times of the longitudinal markers are non-informative~\citep{rizopoulos2011bayesian}, the joint density of $(T_i, \Delta_i)$ can be factorized into a part depending on the distribution of $T_i^\star$ and a part depending on that of $C_i$, so that
\begin{align}\label{eqn:censoring}
f^\star(T_i, \Delta_i| Y_i, g_i = k) &\propto f^\star(T_i| Y_i, g_i = k)^{\Delta_i} S^\star(T_i| Y_i, g_i = k)^{1-\Delta_i}\nonumber\\&=   
\lambda^\star(T_i| Y_i, g_i = k)^{\Delta_i} S^\star(T_i| Y_i, g_i = k),
\end{align}
where $S^\star$ and $\lambda^\star$ are the survival and hazard function associated with the density $f^\star$ of $T_i^\star.$

Under the assumptions given in the previous subsections, all terms in \eqref{eq:log-lik} can be computed in closed form. Indeed, $f_\theta(g_i = k)$ is given by \eqref{eq:pi} and the density function $f_\theta(Y_i | g_i = k)$ can be derived from  the distribution of $b_i$ and \eqref{eq:Y_i_density} (detailed calculations are given in Section~\ref{sec:ext-EM_sup} of the Supplementary Materials). Furthermore, following Equation~\eqref{eqn:censoring}, we have
$f_\theta(T_i, \Delta_i| Y_i, g_i = k) \propto \lambda \big(T_i\,| \, \cY_i(T_i^-), g_i = k\big)^{\Delta_i} S_k(T_i)$,
where $S_k(t) = \exp \Big(-\int_0^{t} \lambda\big(s \, | \, \cY_i(s^-), g_i = k\big) \dd s \Big)$
is the survival function of subject $i$ given that it belongs to latent class $k$. Since the baseline hazard function $\lambda_0$ takes mass only at each failure time $\tau_j \in (\tau_1, \ldots, \tau_J)$ then the integration over the survival process $S_k(t)$ is simply a finite sum over the process evaluated at the $J$ failure times. Then, we rewrite the function $S_k$, for any $t \geq 0$, as
\begin{equation*}
S_k(t) =\exp \Big(-\sum_{j=1}^J \lambda (\tau_j) \ind{\{\tau_j \leq t\}}\Big) = \exp \Big(-\sum_{j=1}^J \lambda_0(\tau_j) \exp \big(\psi_i(\tau_j)^\top \gamma_k\big) \ind{\{\tau_j \leq t\}}\Big).
\end{equation*}
The fact that $f_\theta(T_i, \Delta_i| Y_i, g_i = k)$ is closed-form is one of the major advantages of our model over standard SREMs. Indeed, computing this density in SREMs usually requires integrating it with respect to the distribution of the random effects $b_i$, leading to intractable integrals in the log-likelihood function. These integrals are typically estimated using Monte Carlo techniques~\citep{hickey2018joinerml}, which are computationally intensive and require additional assumptions on the allowed association functions $\psi_i$. These approaches usually do not scale in a high-dimensional context. 

To minimize \eqref{eq:log-lik} with respect to $\theta$, we use the EM algorithm, which is the common choice in the literature~\citep{wulfsohn1997joint, lin2002maximum}. This requires deriving what we call the negative ``complete'' log-likelihood, that is, an estimation of the joint density $f^\ast(T_i, \Delta_i, Y_i, b_i, g_i)$, where the random effect $b_i$ and the latent class $g_i$ are not observed. To this end, we need the following independence assumption.

\begin{assumption}
\label{assumption3}
For any $i \in \{1, \dots, n \}$ and $\ell \in \{1, \ldots, L\}$, the random effects $b^\ell_i$ are independent of the latent class membership $g_i$, and remain independent of it conditionally on 
$T_i$, $\Delta_i$, and $Y_i$.
\end{assumption}

This assumption states that subject-and-longitudinal marker specific random effects $b_i^\ell$ do not depend on the latent class membership.
Then, we have 
\begin{align*}
      f^\star(T_i, \Delta_i, Y_i, b_i, g_i) & =  f^\star(b_i, g_i)f^\star( Y_i | b_i, g_i) f^\star(T_i, \Delta_i | Y_i, b_i, g_i)  \\
    & =f^\star(b_i, g_i)  f^\star( Y_i | b_i, g_i )f^\star(T_i, \Delta_i | Y_i, g_i)
      &(\text{by Assumption \ref{indep-hyp-2}}) \\
      &=f^\star(b_i) f^\star(g_i) f^\star( Y_i | b_i, g_i)f^\star(T_i, \Delta_i | Y_i, g_i).
       & (\text{by Assumption \ref{assumption3}})
\end{align*}
The negative complete log-likelihood is then given by
\begin{align}
\label{eq:comp-log-lik}
  \cL_n^\text{comp}(\theta) 
  = - n^{-1} \sum_{i=1}^n \Big(& \log f_\theta(b_i) +  \sum_{k=1}^{K} \ind{\{g_i=k\}} \big( \log \P_\theta(g_i=k) + \log f_\theta(Y_i \, | \, b_i, g_i = k)\big) \nonumber \\
  &+ \log f_\theta(T_i, \Delta_i \, | \, Y_i, g_i = k) \Big),
\end{align}
 where $f_\theta(b_i)$ is the density of a multivariate gaussian $\cN(0, D)$ distribution and $f_\theta(Y_i \, | \, b_i, g_i = k)$ is typically the density of a $\cN\big(M_{ik},\Sigma_i\big)$ distribution.

\subsection{Penalized objective}
\label{sec:penalized-obj}
To avoid overfitting and provide interpretation on which longitudinal markers are relevant for predicting time-to-event, we propose to minimize the penalized negative log-likelihood
\begin{equation}
  \label{eq:pen-log-lik}
 \cL_n^\text{pen}(\theta) = \cL_n(\theta) + \Omega(\theta) =  \cL_n(\theta) + \sum_{k=1}^{K} \zeta_{1,k} \Omega_1(\xi_k) + \sum_{k=1}^{K} \zeta_{2,k} \Omega_2(\gamma_k),
\end{equation}
where $\Omega_1$ is an elastic net regularization~\citep{zou2005regularization}, 
$\Omega_2$ is a sparse group lasso regularization~\citep{simon2013sparse},
and $(\zeta_{1,k}, \zeta_{2,k})^\top \in (\R^+)^2$  regularization hyperparameters that need to be tuned. 
An advantage of this regularisation strategy is its ability to perform feature selection and to identify the most important features (longitudinal markers and time-independent) relative to the prediction objective. On the one hand, the support of $\xi_k$, controlled by the $\ell_1$ term in $\Omega_1$, provides information about the time-independent features involved in the $k$-th latent class membership while the $\ell_2$ regularization allows to handle correlations between time-independent features.
On the other hand, for the sparse group lasso penalty, a group $\ell$ corresponds to a trajectory, i.e. a longitudinal marker. Thus, if $\gamma_k^\ell$ is completely zero (thanks to the group lasso part), it means that the $\ell$-th longitudinal process is discarded by the model in terms of risk effect for the $k$-th latent class.
Then, the sparse part of the penalty allows a selection of association features for each trajectory: for $\gamma_k^\ell$ that are not completely zeroed, their support informs about the association features involved in the risk of the $k$-th latent class event for the $\ell$-th longitudinal marker.
\subsection{Optimization}
\label{sec:ext-EM}

Given our regularization strategy, we employ an extended version of the EM algorithm \citep{mclachlan2007algorithm} which we now briefly outline. Extensive details on the algorithm are given in Section~\ref{sec:ext-EM_sup} of the Supplementary Materials.  

Our final optimization problem writes
\begin{equation}
  \label{eq:optim-pb}
   \hat \theta \in \underset{\theta \in \R^P}{\argmin}  \, \cL^\text{pen}_n(\theta).
 \end{equation}
Assume that we are at step $w + 1$ of the EM algorithm, with current iterate denoted by $\theta^{(w)}$, then the algorithm consists in the following two steps:
\begin{itemize}
    \item E-step: compute the expected negative complete log-likelihood conditional on the current estimate of the parameters $\theta^{(w)}$, that is, $\cQ_n(\theta, \theta^{(w)}) = \E_{\theta^{(w)}}[\cL_n^\text{comp}(\theta) \,| \, \cD_n]$.
\item M-step: find $\theta^{(w+1)} \in \underset{\theta \in \R^P}{\argmin} \, \cQ^\text{pen}_n(\theta, \theta^{(w)})$,
where $\cQ^\text{pen}_n(\theta, \theta^{(w)}) = \cQ_n(\theta, \theta^{(w)}) + \Omega(\theta)$ and $\Omega(\theta)$ is the penalization defined in \eqref{eq:pen-log-lik}.
\end{itemize}

Under our assumptions, we can show that computing $\cQ_n(\theta, \theta^{(w)})$ reduces to computing the expectations $\E_{\theta^{(w)}}[ b_i | T_i, \Delta_i, Y_i]$ and $\E_{\theta^{(w)}}[ b_i b_i^\top | T_i, \Delta_i, Y_i]$, and the probabilities $\P_{\theta^{(w)}}[g_i = k | T_i, \Delta_i, Y_i]$, $k \in \{1, \dots, K\}$, see Section~\ref{sec:E-step} of the Supplementary Materials for their exact expressions.

Concerning the M-step, we divide the problem into several updates for which we minimize $\cQ^\text{pen}_n(\theta, \theta^{(w)})$ with respect to blocks of coordinates of $\theta$ separately.
The order of these updates matters. The updates for $D, (\beta_k)_{k \in \{1, \dots, K\}}, \lambda_0$, and $\phi$ are easily obtained in closed-form.  The update for $\xi_k^{(w)}$ reduces to the non-smooth convex minimization problem
\begin{equation}
  \label{eq:minimization-xi_k}
  \xi_k^{(w+1)} \in \underset{\xi \in \R^p}{\argmin} \, \mathcal{F}_{1,k}(\xi) + \zeta_{1,k} \Omega_1(\xi),
\end{equation}
where $\mathcal{F}_{1,k}$ is a convex function with respect to $\xi$. Problem~\eqref{eq:minimization-xi_k} is then solved using a quasi-Newton method, the L-BFGS-B algorithm \citep{zhu1997algorithm}. The update for $\gamma_k^{(w)}$ has a similar expression and is solved using proximal gradient descent~\citep{boyd2004convex}. We refer to Section~\ref{sec:ext-EM_sup} of the Supplementary Materials for all details and proofs. The final algorithm is also given with a discussion of its convergence properties.


\section{Evaluation methodology}
\label{sec:evaluation_methodology}

In this section, we present our evaluation strategy to assess real-time prediction performance of our model and briefly introduce the models used for comparison.

\subsection{Real-time prediction and evaluation strategy}
\label{sec:evaluation strategy}


Developments in joint models have focused primarily on modeling and estimation, and most studies do not consider goodness-of-fit or predictive performance of latent class membership or time-to-event~\citep{hickey2016joint}. 
However, for real-time or daily predictions, practitioners need predictive prognostic tools to evaluate and compare survival models. Therefore, we place ourselves in a so-called ``real-time'' prediction setting. Once the learning phase for the model has been completed on a training set, so that one obtains $\hat\theta$ from~\eqref{eq:optim-pb} using the approach described in Section~\ref{sec:ext-EM}, we want to make real-time predictions. More precisely, for each subject $i$, we seek to provide a predictive marker, typically the probability of belonging to a latent class at any time $t$, using all the data available up to that time, but without using the supervision labels $(T_i, \Delta_i)$, which are a priori not available at any time $t$.

\subsubsection{Predictive marker} In our setting, since each latent class represents the different risk levels of a subject, we choose the probability of latent class membership as the predictive marker. This is similar to what is classically done in JLCMs, where $\widetilde{\pi}_{ik}^{\hat\theta}= \P_{\hat\theta}[g_i = k | T_i, \Delta_i, Y_i]$ is typically used as the predictive rule \citep[see, e.g.,][]{proust2014joint}. However, this requires knowledge of the survival labels $(T_i, \Delta_i)$, which does not fit in our real-time prediction goal. Therefore, we define a new predictive marker as follows. 

For any subject $i$ and any time $s_i$ elapsed since entry into the study, given longitudinal markers $\cY_i(s_i^{-})$ observed up to $s_i$, for any $k \in \{1, \dots, K \}$, we let 
\[ \widehat \cR_{ik}(s_i) = \P_{\hat \theta}\big[g_i=k \, | \, T^\star_i > s_i, \cY_i(s_i^{-}) \big].\]
Indeed, for any subject $i$ who is event-free when $s_i$ has elapsed, all we know about that subject is that its time to the event of interest $T_i^\star$ exceeds $s_i$. This is equivalent to considering this subject as a new subject for which $T_i = s_i$, $\Delta_i = 0$, and $Y_i = \cY_i(s_i^{-})$. The expression of $\widehat \cR_{ik}(s_i)$ can thus be rewritten as $\widehat \cR_{ik}(s_i) = \P_{\hat \theta}[g_i=k \, | \, T_i = s_i, \Delta_i = 0, Y_i = \cY_i(s_i^{-})],$ see Lemma \ref{lemma:E_mu_b_i_gi} of the Supplementary Materials for more details.


We illustrate this real-time prediction setting in Figure \ref{fig:real_time}, emphasizing that $s_i$ should be thought of as the duration between the enrollment of individual $i$ and the ``present'' time.

\begin{figure}
    \centering
    \includegraphics[scale=.4]{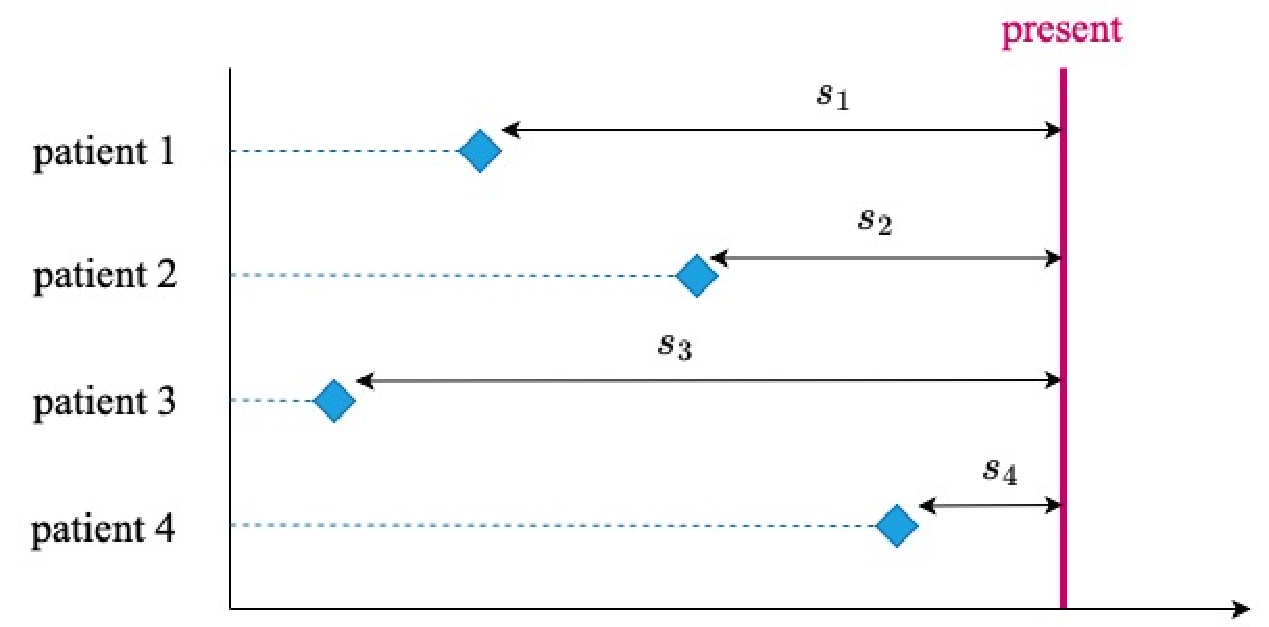}
    \caption{Real-time prediction setting. In a practical application, we want to be able to make predictions at any ``present'' time while subjects have entered the study at different times. Therefore, some of them have a lot of recorded information while the others have a few. }
    \label{fig:real_time}
\end{figure}

\subsubsection{Performance evaluation} 
We want to compare the quality of our predictions to the true labels $(T_i, \Delta_i)$, to which we have access in these comparison experiments. Given a test set in which each individual's trajectory is fully observed until the end of the study, we mimic the real-time prediction setting by randomly sampling the $s_i$.

We use the classical C-index~\citep{harrell1996tutorial} as our performance metric. More precisely, we assume that we are in the case $K=2$ and the class $g_i=2$ represents the high-risk group of subjects (the class $g_i=1$ then representing a low-risk group). We denote by $\widehat \cR_i = \widehat \cR_{i2}(s_i)$ the predictive marker that a subject $i$ belongs to class $g_i=2$ when $s_i$ has elapsed. Then, we let
$\cC =\P[\widehat \cR_i > \widehat \cR_j \, | \, T^\star_i < T^\star_j],
$ with $i \neq j$ two random independent subjects (note that $\cC$ does not depend on $i, j$ under the i.i.d.~sample hypothesis). 

In our case,  $T^\star$ is subject to right censoring, so one would typically consider the modified $\widetilde{\cC}$ defined by
$\widetilde{\cC} =\P[\widehat \cR_i > \widehat \cR_j \, | \, T_i < T_j , \, T_i < t^{\textnormal{max}}],$ where $t^{\textnormal{max}}$ corresponds to a fixed and predetermined follow-up period~\citep{heagerty2005survival}. It has been shown by \citet{uno2011c} that a Kaplan-Meier estimator for the censoring distribution leads to a nonparametric and consistent estimator of $\widetilde{\cC}$. 

In Section~\ref{sec:evaluation-procedure} of the Supplementary Materials, we give the complete procedure used to evaluate the performance of the models considered in our experiments.

\subsection{Competing models}
\label{sec:competing models}




We compare FLASH with the very classical and widely used 
\texttt{LCMM}~\citep{2017_lcmm} and \texttt{JMbayes}~\citet{2017_JMBayes}, which are extensions of JLCM and SREM  that allow for multivariate longitudinal markers. We present them, along with their respective predictive markers, in Section~\ref{sec:SREM_JLCM} of the Supplementary Materials. Note that not many joint models allow for multivariate longitudinal markers, which limits our choice of competing methods.

\section{Experimental results}
\label{sec:experiment}

To evaluate our method, we first perform in Subsection \ref{sec:simu-1} a simulation study that illustrates our estimation procedure. We then turn to a comparison study on both simulated and medical examples in Subsection \ref{sec:compare-data}, 
and finally show in Subsection \ref{sec:medical-data} that the biomarkers identified as significant by our model are consistent with current medical knowledge.

In all experiments, the features extracted by the \texttt{tsfresh} package~\citep{christ2018time} are used for association features $\Psi_m$ in FLASH. This package extracts dozens of features from a time series such as absolute energy, kurtosis, or autocorrelation. Before running our extended EM algorithm with the set of features extracted by the \texttt{tsfresh} package, we use a screening phase procedure where we select the top ten association features by fitting the extracted feature of each candidate and the survival labels in individual Cox models and comparing their C-index scores. In addition, a recent line of work is to use the signature transform~\citep{fermanian2021embedding, bleistein2024dynamical} to extract features from longitudinal markers. This transform encapsulates geometric information about multivariate time series. We provide additional results with the signature transform in Section~\ref{sec:sig_transform} of the Supplementary Materials, which show that our method is generic and performs well regardless of the feature extraction functions used.

We tune the regularization hyperparameters $(\zeta_{1,k}, \zeta_{2,k})_{k \in \{1, \dots, K\}}$ with a grid search and a 10-fold cross-validation with the C-index metric. Note that we keep $\zeta_{1,1} = \cdots = \zeta_{1,K}$ and $\zeta_{2,1} = \cdots = \zeta_{2,K}$. Extensive details on our experiments together with additional results on a high-dimensional dataset from NASA are given in Section~\ref{sec:experiment_detail} of the Supplementary Materials.

\subsection{Simulation study}\label{sec:simu-1}

\begin{figure}[]
    \centering
    \includegraphics[width=.7\textwidth]{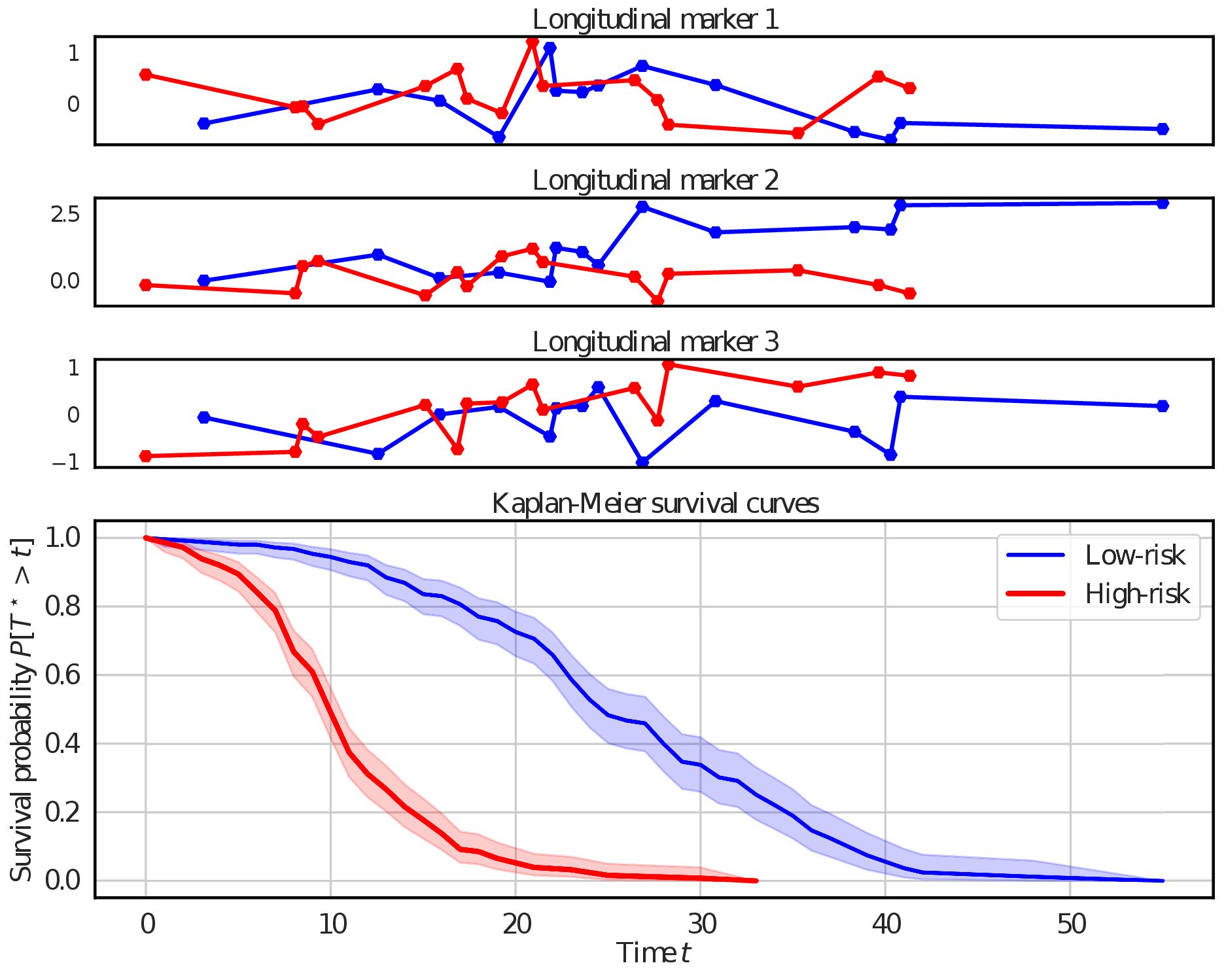}
    \caption{Simulated cohort of $n=500$ samples for $K=2$ groups (high-risk group in red curves and low-risk group in blue curves). Top figures:  trajectories of first three longitudinal markers of two individuals randomly selected in each group. Bottom figure: Kaplan-Meier survival curves for each group.}
    \label{fig:simu_data}
\end{figure}

To assess our estimation procedure, we simulate data as follows. First, we simulate the latent class membership indicator from the logistic regression model in~\eqref{eq:pi}. Based on this indicator, we divide the population into two groups: a high-risk group and a low-risk group. Within each group, we apply the classical survival simulation setting described by~\cite{bender2005generating}.  Next, we generate the longitudinal markers 
 $Y$ from the generalized linear mixed models in~\eqref{eq:Y_i_density}. Survival times are generated from their hazard functions given in~\eqref{eq:intensity-model}. The model coefficients $\xi_1$, $\xi_2$, and $\gamma$ of models~\eqref{eq:pi} and~\eqref{eq:intensity-model} are generated as sparse vectors, ensuring that only a subset of the corresponding features are active (i.e., the coefficients are non-zero). Extensive details of these simulations are provided in Section~\ref{sec:data-simu} of the Supplementary Materials. Figure \ref{fig:simu_data} shows some examples of simulated longitudinal markers and Kaplan-Meier survival curves.

\begin{figure}[]
    \centering
    \includegraphics[width=.7\textwidth]{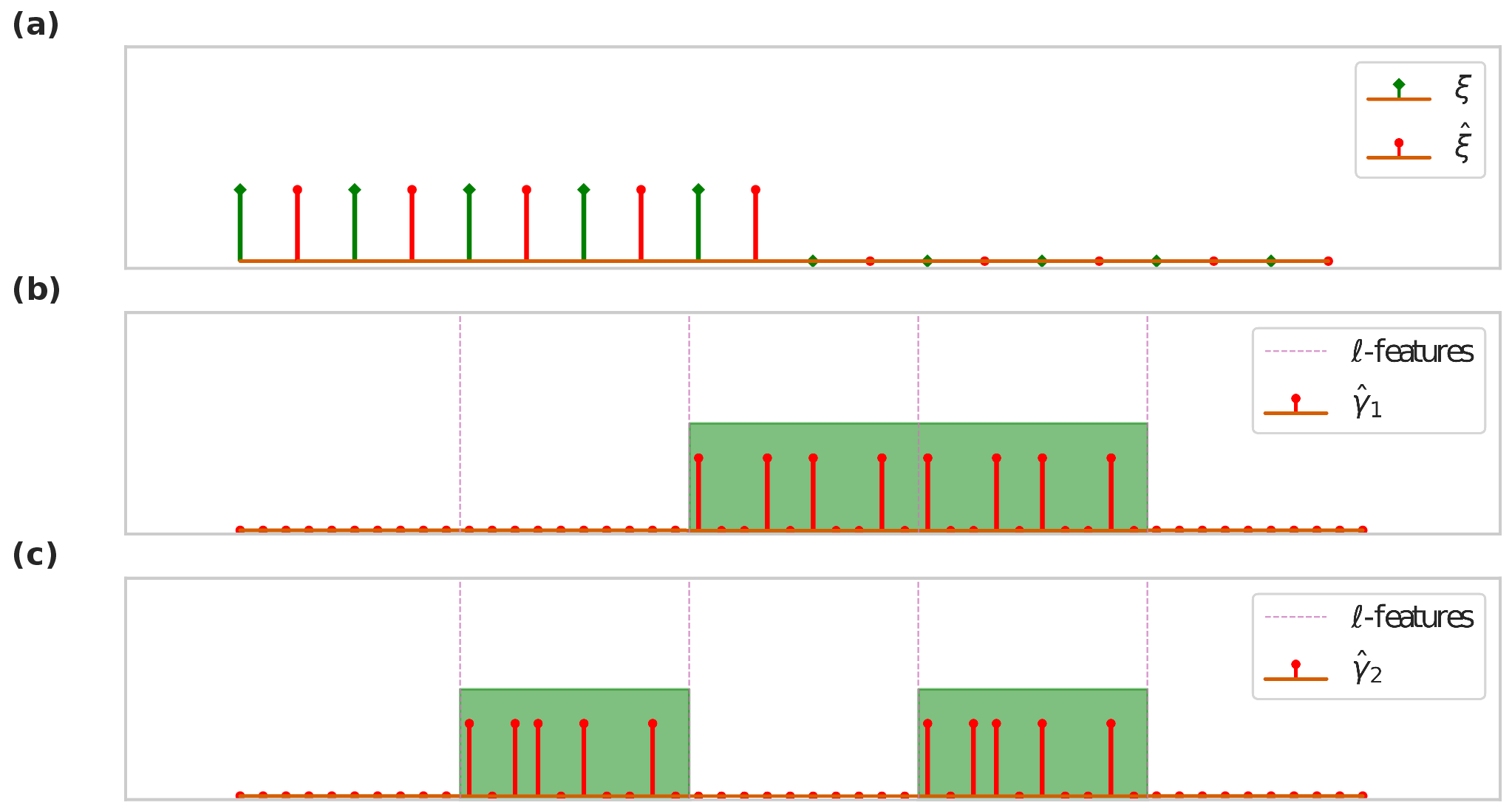}
    \caption{Simulations results. \textbf{(a)}: the support of both the true coefficient $\xi$ in green and its estimated version $\hat{\xi}$ in red. \textbf{(b)} and \textbf{(c)}: in red the support of the estimated coefficient $\hat{\gamma}_k$ for $k \in \{1, 2 \}$, the dashed pink lines separate the features corresponding to each longitudinal marker $\ell$, and active longitudinal markers are represented by a green area.}
    \label{fig:reg_eff}
\end{figure}

To illustrate our regularization strategy, we show in Figure \ref{fig:reg_eff} the time-independent parameter $\xi$ and the joint association parameters $(\gamma_k)_{k \in \{1, 2 \}}$ and their estimation obtained after running our learning procedure. We see in the sub-figure \textbf{(a)} that the support of $\xi$ is fully recovered thanks to the elastic-net penalty.
Additionally, sub-figures \textbf{(b)} and \textbf{(c)} demonstrate the effect of the sparse group lasso, showing that only the coefficients corresponding to active longitudinal features (represented by the green area) are non-zero, while all coefficients for inactive longitudinal features are zero. 

\subsection{Comparison study}
\label{sec:compare-data}

We compare FLASH with \texttt{JMbayes} and \texttt{LCMM} on two simulated and two real-world datasets and use the C-index metric presented in Section \ref{sec:evaluation_methodology}. The first simulated dataset is the one from the previous subsection, and the second one is from the R package \citep{hickey2018joinerml} \texttt{joineRML}. We give detailed description of this second simulated dataset and the two real-world datasets (\textit{PBCseq} and \textit{Aids}) in Section~\ref{sec:compared_data} of the Supplementary Materials. A summary of the datasets is given in Table \ref{tab:summary_dataset}.

\begin{table}[h]

\caption{Datasets characteristics: the number of samples $n$, the number of longitudinal features $L$, number of time-independent features $p$, and the overall number of parameters in FLASH model $P$. The names FLASH\_simu and joineRML\_simu correspond to the datasets simulated from the simulation study in Section \ref{sec:simu-1} and the \texttt{joineRML} package respectively.}
\label{tab:summary_dataset}
\begin{center}
\begin{tabular}{ccccc} 
 \toprule
 \textbf{Dataset} & $n$ & $L$ & $p$& $P$ \\ 
 \midrule
 FLASH\_simu & 500 & 5 & 10 & 224\\
 joineRML\_simu & 250 & 2 & 2 & 204\\ 
 PBCseq & 304 & 7 & 3 & 251 \\
 Aids & 467 & 1 & 4 & 147 \\ 
 Sepsis & 654 & 4 & 21 & 255\\
 \bottomrule
\end{tabular}
\end{center}
\end{table}

\begin{figure}[!htb]
    \centering
    \includegraphics[width=.7\textwidth]{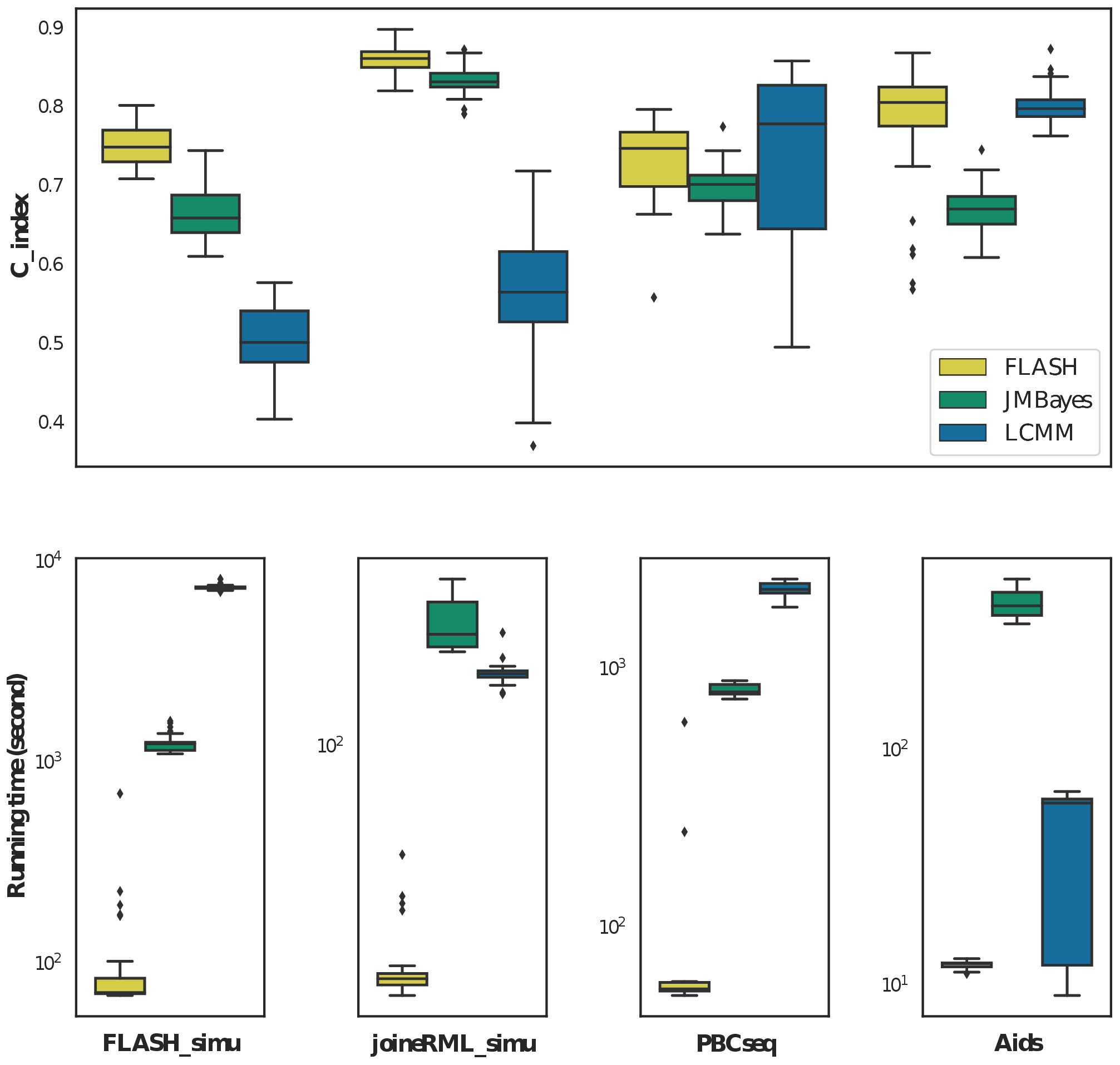}
    \caption{C-index (\textbf{top} figure) and runtime (\textbf{bottom} figures) comparison on the four datasets considered. The box plots of C-index and runtime are obtained with 50 independent experiments.}
    \label{fig:real_data_perf}
\end{figure}

We can see in Figure~\ref{fig:real_data_perf} that FLASH outperforms its competitors in terms of both C-index and running times on all datasets. The good performance of FLASH in terms of running times can be explained by the fact that it does not need to perform computationally intensive Monte Carlo techniques like \texttt{JMbayes}, while it is easier to satisfy the convergence criterion of our EM algorithm than that of \texttt{LCMM}.

\subsection{Biological interpretation of FLASH results} \label{sec:medical-data}
{Considering the growing emphasis on model interpretability and the fact that the new regulations in the European Union and the United States now require that a model be interpretable and understandable to be certified as a medical device \citep{geller2023food,panigutti2023role}, we conclude this section with interpretations of the results of FLASH on the PBCseq and Sepsis datasets.} The Sepsis dataset 
describes the sepsis diagnosis of patients, where, after a pre-processing step, 4 multivariate longitudinal features and 21 time-independent features are available for each patient. The coefficients estimated by FLASH, and in particular their sparsity, provide us information on which marker is involved in the diagnosis, see Section~\ref{sec:medical} of the Supplementary Materials for all numerical values. Note that we did not include the Sepsis dataset in Section~\ref{sec:compare-data} because both \texttt{LCMM} and \texttt{JMBayes} take a very long time to converge, so that we could not do the Monte Carlo comparison. This highlights the fact that they do not scale to high-dimensional settings. However, on one experimental trial, \texttt{LCMM} and \texttt{JMBayes} only give C-index scores of 0.51 and 0.56 while FLASH has a score of 0.74.


For the PBCseq dataset, alkaline phosphatase appears to be the most important variable, followed by prothrombin and albumin. Alkanine phosphatase is already recognized to be an important variable to monitor: it is already known that phosphatase alkaline at 6-month predicts non-responders and survival \citep{perez2023optimizing}, it is recognized that treatment target should be normalization of alkaline phosphatase \citep{perez2020goals}, and prognosis is improved for patients taking drugs lowering alkaline phosphatase as ursodeoxycholic acid \citep{kuiper2009improved}. Concerning the Sepsis dataset, respiratory rate is the most important variable. Systolic blood pressure appears to be the most important longitudinal feature, followed by oxygen saturation. Of note, the two most widely used prognostic criteria in sepsis, i.e. the qSOFA and the SIRS criteria, which contain respectively 3 and 4 variables, both include among these variables respiratory rate \citep{raith2017prognostic}.

\section{Discussion}
\label{sec:conclusion}

In this paper, a generalized joint model for high-dimensional multivariate longitudinal data and censored durations (FLASH) has been introduced, with an efficient estimation methodology based on an extension of the EM algorithm. This algorithm allows the use of regularization strategies in order to perform feature selection and results in an interpretable model scalable to high-dimensional longitudinal markers. We evaluated the performance of the estimation procedure on an extensive Monte Carlo simulation study. It showed that our method successfully recovered the most significant features. The proposed methodology has then been applied on four different datasets. On these datasets, FLASH outperforms competing methods, both in terms of C-index and runtimes, in a so-called ``real-time'' prediction setting. In addition, we show on experiments of medical datasets that our model automatically identifies the most important longitudinal markers and time-independent features, allowing important interpretations on the application at hand. Potential future work consists of extending the implementation to generalize our EM algorithm to support count or binary longitudinal features, and to relax the assumptions on latent class membership to allow classes to change with time, using for example a Markov structure.

\backmatter


\section*{Acknowledgements}
\textit{Conflict of Interest}: None declared. For this work, Agathe Guilloux and Van Tuan Nguyen benefited from the support of the Agence Nationale pour la Recherche under the France 2030 program with the reference ANR-22-PESN-0016. This work was also supported by the French National Cancer Institut (INCa) [grant number 2016-1-PL SHS-03-1]. The authors thank Linus Bleistein and Massil Hihat for fruitful discussions. 


\section*{Supplementary Materials}

Supplementary Materials, which includes the sections referenced in the main paper, is available with this paper at the Biometrics website on Wiley Online Library.\vspace*{-8pt}


\bibliographystyle{biom} \bibliography{refs}

\begin{thebibliography}{}

\bibitem[\protect\citeauthoryear{Aldrich and Nelson}{Aldrich and Nelson}{1984}]{aldrich1984linear}
Aldrich, J.~H. and Nelson, F.~D. (1984).
\newblock {\em Linear probability, logit, and probit models}.
\newblock Number~45. Sage.

\bibitem[\protect\citeauthoryear{Andersen and Gill}{Andersen and Gill}{1982}]{andersen1982cox}
Andersen, P.~K. and Gill, R.~D. (1982).
\newblock Cox's regression model for counting processes: a large sample study.
\newblock {\em The annals of statistics} pages 1100--1120.

\bibitem[\protect\citeauthoryear{Andrew and Gao}{Andrew and Gao}{2007}]{andrew2007scalable}
Andrew, G. and Gao, J. (2007).
\newblock Scalable training of l1-regularized log-linear models.
\newblock In {\em International Conference on Machine Learning}, pages 33--40. ACM.

\bibitem[\protect\citeauthoryear{Andrinopoulou, Nasserinejad, Szczesniak, and Rizopoulos}{Andrinopoulou et~al.}{2020}]{andrinopoulou2020integrating}
Andrinopoulou, E.-R., Nasserinejad, K., Szczesniak, R., and Rizopoulos, D. (2020).
\newblock Integrating latent classes in the bayesian shared parameter joint model of longitudinal and survival outcomes.
\newblock {\em Statistical methods in medical research} {\bf 29,} 3294--3307.

\bibitem[\protect\citeauthoryear{Andrinopoulou and Rizopoulos}{Andrinopoulou and Rizopoulos}{2016}]{andrinopoulou2016bayesian}
Andrinopoulou, E.-R. and Rizopoulos, D. (2016).
\newblock Bayesian shrinkage approach for a joint model of longitudinal and survival outcomes assuming different association structures.
\newblock {\em Statistics in medicine} {\bf 35,} 4813--4823.

\bibitem[\protect\citeauthoryear{Austin}{Austin}{2013}]{austin2013correction}
Austin, P.~C. (2013).
\newblock Correction:‘generating survival times to simulate cox proportional hazards models with time-varying covariates’.
\newblock {\em Statistics in Medicine} {\bf 32,} 1078--1078.

\bibitem[\protect\citeauthoryear{Bach, Jenatton, Mairal, Obozinski, et~al\mbox{.}}{Bach et~al.}{2012}]{bach2012optimization}
Bach, F., Jenatton, R., Mairal, J., Obozinski, G., et~al. (2012).
\newblock Optimization with sparsity-inducing penalties.
\newblock {\em Foundations and Trends{\textregistered} in Machine Learning} {\bf 4,} 1--106.

\bibitem[\protect\citeauthoryear{Bach}{Bach}{2008}]{bach2008bolasso}
Bach, F.~R. (2008).
\newblock Bolasso: model consistent lasso estimation through the bootstrap.
\newblock In {\em Proceedings of the 25th international conference on Machine learning}, pages 33--40.

\bibitem[\protect\citeauthoryear{Bartolucci and Farcomeni}{Bartolucci and Farcomeni}{2015}]{bartolucci2015discrete}
Bartolucci, F. and Farcomeni, A. (2015).
\newblock A discrete time event-history approach to informative drop-out in mixed latent markov models with covariates.
\newblock {\em Biometrics} {\bf 71,} 80--89.

\bibitem[\protect\citeauthoryear{Bartolucci and Farcomeni}{Bartolucci and Farcomeni}{2019}]{bartolucci2019shared}
Bartolucci, F. and Farcomeni, A. (2019).
\newblock A shared-parameter continuous-time hidden markov and survival model for longitudinal data with informative dropout.
\newblock {\em Statistics in medicine} {\bf 38,} 1056--1073.

\bibitem[\protect\citeauthoryear{Beck and Teboulle}{Beck and Teboulle}{2009}]{beck2009fast}
Beck, A. and Teboulle, M. (2009).
\newblock A fast iterative shrinkage-thresholding algorithm for linear inverse problems.
\newblock {\em SIAM journal on imaging sciences} {\bf 2,} 183--202.

\bibitem[\protect\citeauthoryear{Bender, Augustin, and Blettner}{Bender et~al.}{2005}]{bender2005generating}
Bender, R., Augustin, T., and Blettner, M. (2005).
\newblock Generating survival times to simulate cox proportional hazards models.
\newblock {\em Statistics in medicine} {\bf 24,} 1713--1723.

\bibitem[\protect\citeauthoryear{Bleistein, Nguyen, Fermanian, and Guilloux}{Bleistein et~al.}{2024}]{bleistein2024dynamical}
Bleistein, L., Nguyen, V.-T., Fermanian, A., and Guilloux, A. (2024).
\newblock Dynamical survival analysis with controlled latent states.
\newblock {\em 41th International Conference on Machine Learning and arXiv preprint arXiv:2401.17077} .

\bibitem[\protect\citeauthoryear{B{\"o}hning}{B{\"o}hning}{1992}]{bohning1992multinomial}
B{\"o}hning, D. (1992).
\newblock Multinomial logistic regression algorithm.
\newblock {\em Annals of the institute of Statistical Mathematics} {\bf 44,} 197--200.

\bibitem[\protect\citeauthoryear{Boyd and Vandenberghe}{Boyd and Vandenberghe}{2004}]{boyd2004convex}
Boyd, S. and Vandenberghe, L. (2004).
\newblock {\em Convex optimization}.
\newblock Cambridge university press, New York.

\bibitem[\protect\citeauthoryear{Breslow}{Breslow}{1972}]{breslow1972contribution}
Breslow, N.~E. (1972).
\newblock Contribution to the discussion of the paper by dr cox.
\newblock {\em Journal of the Royal Statistical Society, Series B} {\bf 34,} 216--217.

\bibitem[\protect\citeauthoryear{Bussy, Guilloux, Ga{\"\i}ffas, and Jannot}{Bussy et~al.}{2019}]{bussy2019c}
Bussy, S., Guilloux, A., Ga{\"\i}ffas, S., and Jannot, A.-S. (2019).
\newblock C-mix: A high-dimensional mixture model for censored durations, with applications to genetic data.
\newblock {\em Statistical methods in medical research} {\bf 28,} 1523--1539.

\bibitem[\protect\citeauthoryear{Chi and Ibrahim}{Chi and Ibrahim}{2006}]{chi2006joint}
Chi, Y.-Y. and Ibrahim, J.~G. (2006).
\newblock Joint models for multivariate longitudinal and multivariate survival data.
\newblock {\em Biometrics} {\bf 62,} 432--445.

\bibitem[\protect\citeauthoryear{Christ, Braun, Neuffer, and Kempa-Liehr}{Christ et~al.}{2018}]{christ2018time}
Christ, M., Braun, N., Neuffer, J., and Kempa-Liehr, A.~W. (2018).
\newblock Time series feature extraction on basis of scalable hypothesis tests (tsfresh--a python package).
\newblock {\em Neurocomputing} {\bf 307,} 72--77.

\bibitem[\protect\citeauthoryear{Chzhen, Hebiri, and Salmon}{Chzhen et~al.}{2019}]{chzhen2019lasso}
Chzhen, E., Hebiri, M., and Salmon, J. (2019).
\newblock On lasso refitting strategies.

\bibitem[\protect\citeauthoryear{Cox}{Cox}{1972}]{Cox1972JRSS}
Cox, D.~R. (1972).
\newblock Regression models and life-tables.
\newblock {\em Journal of the Royal Statistical Society. Series B (Methodological)} {\bf 34,} 187--220.

\bibitem[\protect\citeauthoryear{Dempster, Laird, and Rubin}{Dempster et~al.}{1977}]{dempster1977maximum}
Dempster, A., Laird, N., and Rubin, D. (1977).
\newblock Maximum likelihood from incomplete data via the em algorithm.
\newblock {\em Journal of the Royal Statistical Society. Series B (Methodological)} {\bf 39,} 1--38.

\bibitem[\protect\citeauthoryear{Devaux, Genuer, Peres, and Proust-Lima}{Devaux et~al.}{2022}]{devaux2022individual}
Devaux, A., Genuer, R., Peres, K., and Proust-Lima, C. (2022).
\newblock Individual dynamic prediction of clinical endpoint from large dimensional longitudinal biomarker history: a landmark approach.
\newblock {\em BMC Medical Research Methodology} {\bf 22,} 1--14.

\bibitem[\protect\citeauthoryear{Fabian~Pedregosa}{Fabian~Pedregosa}{2020}]{copt}
Fabian~Pedregosa, Geoffrey~Negiar, G.~D. (2020).
\newblock copt: composite optimization in python.

\bibitem[\protect\citeauthoryear{Fabio, Paula, and de~Castro}{Fabio et~al.}{2012}]{fabio2012poisson}
Fabio, L.~C., Paula, G.~A., and de~Castro, M. (2012).
\newblock A poisson mixed model with nonnormal random effect distribution.
\newblock {\em Computational Statistics \& Data Analysis} {\bf 56,} 1499--1510.

\bibitem[\protect\citeauthoryear{Fermanian}{Fermanian}{2021}]{fermanian2021embedding}
Fermanian, A. (2021).
\newblock Embedding and learning with signatures.
\newblock {\em Computational Statistics \& Data Analysis} {\bf 157,} 107148.

\bibitem[\protect\citeauthoryear{Fitzmaurice, Laird, and Ware}{Fitzmaurice et~al.}{2012}]{fitzmaurice2012applied}
Fitzmaurice, G.~M., Laird, N.~M., and Ware, J.~H. (2012).
\newblock {\em Applied longitudinal analysis}.
\newblock John Wiley \& Sons, New Jersey.

\bibitem[\protect\citeauthoryear{Geller}{Geller}{2023}]{geller2023food}
Geller, J. (2023).
\newblock Food and drug administration published final guidance on clinical decision support software.
\newblock {\em Journal of Clinical Engineering} {\bf 48,} 3--7.

\bibitem[\protect\citeauthoryear{Gompertz}{Gompertz}{1825}]{gompertz1825xxiv}
Gompertz, B. (1825).
\newblock Xxiv. on the nature of the function expressive of the law of human mortality, and on a new mode of determining the value of life contingencies. in a letter to francis baily, esq. frs \&c.
\newblock {\em Philosophical transactions of the Royal Society of London} pages 513--583.

\bibitem[\protect\citeauthoryear{Harrell, Lee, and Mark}{Harrell et~al.}{1996}]{harrell1996tutorial}
Harrell, F.~E., Lee, K.~L., and Mark, D.~B. (1996).
\newblock Tutorial in biostatistics multivariable prognostic models: issues in developing models, evaluating assumptions and adequacy, and measuring and reducing errors.
\newblock {\em Statistics in medicine} {\bf 15,} 361--387.

\bibitem[\protect\citeauthoryear{Hastie, Tibshirani, Friedman, and Friedman}{Hastie et~al.}{2009}]{hastie2009elements}
Hastie, T., Tibshirani, R., Friedman, J.~H., and Friedman, J.~H. (2009).
\newblock {\em The elements of statistical learning: data mining, inference, and prediction}, volume~2.
\newblock Springer.

\bibitem[\protect\citeauthoryear{Hatfield, Boye, and Carlin}{Hatfield et~al.}{2011}]{hatfield2011joint}
Hatfield, L.~A., Boye, M.~E., and Carlin, B.~P. (2011).
\newblock Joint modeling of multiple longitudinal patient-reported outcomes and survival.
\newblock {\em Journal of Biopharmaceutical Statistics} {\bf 21,} 971--991.

\bibitem[\protect\citeauthoryear{Heagerty and Zheng}{Heagerty and Zheng}{2005}]{heagerty2005survival}
Heagerty, P.~J. and Zheng, Y. (2005).
\newblock Survival model predictive accuracy and roc curves.
\newblock {\em Biometrics} {\bf 61,} 92--105.

\bibitem[\protect\citeauthoryear{Hickey, Philipson, Jorgensen, and Kolamunnage-Dona}{Hickey et~al.}{2016}]{hickey2016joint}
Hickey, G.~L., Philipson, P., Jorgensen, A., and Kolamunnage-Dona, R. (2016).
\newblock Joint modelling of time-to-event and multivariate longitudinal outcomes: recent developments and issues.
\newblock {\em BMC medical research methodology} {\bf 16,} 117.

\bibitem[\protect\citeauthoryear{Hickey, Philipson, Jorgensen, and Kolamunnage-Dona}{Hickey et~al.}{2018}]{hickey2018joinerml}
Hickey, G.~L., Philipson, P., Jorgensen, A., and Kolamunnage-Dona, R. (2018).
\newblock joinerml: a joint model and software package for time-to-event and multivariate longitudinal outcomes.
\newblock {\em BMC medical research methodology} {\bf 18,} 1--14.

\bibitem[\protect\citeauthoryear{Jaffa, Gebregziabher, and Jaffa}{Jaffa et~al.}{2014}]{jaffa2014joint}
Jaffa, M.~A., Gebregziabher, M., and Jaffa, A.~A. (2014).
\newblock A joint modeling approach for right censored high dimensional multivariate longitudinal data.
\newblock {\em Journal of biometrics \& biostatistics} {\bf 5,}.

\bibitem[\protect\citeauthoryear{Kang and Song}{Kang and Song}{2022}]{kang2022consistent}
Kang, K. and Song, X. (2022).
\newblock Consistent estimation of a joint model for multivariate longitudinal and survival data with latent variables.
\newblock {\em Journal of Multivariate Analysis} {\bf 187,} 104827.

\bibitem[\protect\citeauthoryear{Klein}{Klein}{1992}]{klein1992semiparametric}
Klein, J.~P. (1992).
\newblock Semiparametric estimation of random effects using the cox model based on the em algorithm.
\newblock {\em Biometrics} pages 795--806.

\bibitem[\protect\citeauthoryear{Klein and Moeschberger}{Klein and Moeschberger}{2005}]{klein2005survival}
Klein, J.~P. and Moeschberger, M.~L. (2005).
\newblock {\em Survival analysis: techniques for censored and truncated data}.
\newblock Springer Science \& Business Media, New York.

\bibitem[\protect\citeauthoryear{Kuiper, Hansen, de~Vries, den Ouden-Muller, Van~Ditzhuijsen, Haagsma, Houben, Witteman, van Erpecum, van Buuren, et~al\mbox{.}}{Kuiper et~al.}{2009}]{kuiper2009improved}
Kuiper, E.~M., Hansen, B.~E., de~Vries, R.~A., den Ouden-Muller, J.~W., Van~Ditzhuijsen, T.~J., Haagsma, E.~B., Houben, M.~H., Witteman, B.~J., van Erpecum, K.~J., van Buuren, H.~R., et~al. (2009).
\newblock Improved prognosis of patients with primary biliary cirrhosis that have a biochemical response to ursodeoxycholic acid.
\newblock {\em Gastroenterology} {\bf 136,} 1281--1287.

\bibitem[\protect\citeauthoryear{Lin, McCulloch, and Mayne}{Lin et~al.}{2002}]{lin2002maximum}
Lin, H., McCulloch, C.~E., and Mayne, S.~T. (2002).
\newblock Maximum likelihood estimation in the joint analysis of time-to-event and multiple longitudinal variables.
\newblock {\em Statistics in Medicine} {\bf 21,} 2369--2382.

\bibitem[\protect\citeauthoryear{Lin, Turnbull, McCulloch, and Slate}{Lin et~al.}{2002}]{lin2002latent}
Lin, H., Turnbull, B.~W., McCulloch, C.~E., and Slate, E.~H. (2002).
\newblock Latent class models for joint analysis of longitudinal biomarker and event process data: application to longitudinal prostate-specific antigen readings and prostate cancer.
\newblock {\em Journal of the American Statistical Association} {\bf 97,} 53--65.

\bibitem[\protect\citeauthoryear{Masoudnia and Ebrahimpour}{Masoudnia and Ebrahimpour}{2014}]{masoudnia2014mixture}
Masoudnia, S. and Ebrahimpour, R. (2014).
\newblock Mixture of experts: a literature survey.
\newblock {\em Artificial Intelligence Review} {\bf 42,} 275--293.

\bibitem[\protect\citeauthoryear{McLachlan and Krishnan}{McLachlan and Krishnan}{2007}]{mclachlan2007algorithm}
McLachlan, G.~J. and Krishnan, T. (2007).
\newblock {\em The EM algorithm and extensions}.
\newblock John Wiley \& Sons.

\bibitem[\protect\citeauthoryear{Molenberghs, Verbeke, Dem{\'e}trio, and Vieira}{Molenberghs et~al.}{2010}]{molenberghs2010family}
Molenberghs, G., Verbeke, G., Dem{\'e}trio, C.~G., and Vieira, A.~M. (2010).
\newblock A family of generalized linear models for repeated measures with normal and conjugate random effects.

\bibitem[\protect\citeauthoryear{Moreau}{Moreau}{1962}]{moreau1962fonctions}
Moreau, J.~J. (1962).
\newblock Fonctions convexes duales et points proximaux dans un espace hilbertien.
\newblock {\em Comptes rendus hebdomadaires des s{\'e}ances de l'Acad{\'e}mie des sciences} {\bf 255,} 2897--2899.

\bibitem[\protect\citeauthoryear{Mukherjee and Maiti}{Mukherjee and Maiti}{1988}]{mukherjee1988some}
Mukherjee, B.~N. and Maiti, S.~S. (1988).
\newblock On some properties of positive definite toeplitz matrices and their possible applications.
\newblock {\em Linear algebra and its applications} {\bf 102,} 211--240.

\bibitem[\protect\citeauthoryear{Murray and Philipson}{Murray and Philipson}{2022}]{murray2022fast}
Murray, J. and Philipson, P. (2022).
\newblock A fast approximate em algorithm for joint models of survival and multivariate longitudinal data.
\newblock {\em Computational Statistics \& Data Analysis} {\bf 170,} 107438.

\bibitem[\protect\citeauthoryear{Panigutti, Hamon, Hupont, Fernandez~Llorca, Fano~Yela, Junklewitz, Scalzo, Mazzini, Sanchez, Soler~Garrido, et~al\mbox{.}}{Panigutti et~al.}{2023}]{panigutti2023role}
Panigutti, C., Hamon, R., Hupont, I., Fernandez~Llorca, D., Fano~Yela, D., Junklewitz, H., Scalzo, S., Mazzini, G., Sanchez, I., Soler~Garrido, J., et~al. (2023).
\newblock The role of explainable ai in the context of the ai act.
\newblock In {\em Proceedings of the 2023 ACM Conference on Fairness, Accountability, and Transparency}, pages 1139--1150.

\bibitem[\protect\citeauthoryear{Perez, Harms, Lindor, Van~Buuren, Hirschfield, Corpechot, Van Der~Meer, Feld, Gulamhusein, Lammers, et~al\mbox{.}}{Perez et~al.}{2020}]{perez2020goals}
Perez, C. F.~M., Harms, M.~H., Lindor, K.~D., Van~Buuren, H.~R., Hirschfield, G.~M., Corpechot, C., Van Der~Meer, A.~J., Feld, J.~J., Gulamhusein, A., Lammers, W.~J., et~al. (2020).
\newblock Goals of treatment for improved survival in primary biliary cholangitis: treatment target should be bilirubin within the normal range and normalization of alkaline phosphatase.
\newblock {\em Official journal of the American College of Gastroenterology| ACG} {\bf 115,} 1066--1074.

\bibitem[\protect\citeauthoryear{Perez, Ioannou, Hassanally, Trivedi, Corpechot, van~der Meer, Lammers, Battezzati, Lindor, Nevens, et~al\mbox{.}}{Perez et~al.}{2023}]{perez2023optimizing}
Perez, C. F.~M., Ioannou, S., Hassanally, I., Trivedi, P.~J., Corpechot, C., van~der Meer, A.~J., Lammers, W.~J., Battezzati, P.~M., Lindor, K.~D., Nevens, F., et~al. (2023).
\newblock Optimizing therapy in primary biliary cholangitis: Alkaline phosphatase at six months identifies one-year non-responders and predicts survival.
\newblock {\em Liver international: official journal of the International Association for the Study of the Liver} {\bf 43,} 1497--1506.

\bibitem[\protect\citeauthoryear{Proust-Lima, Philipps, and Liquet}{Proust-Lima et~al.}{2017}]{2017_lcmm}
Proust-Lima, C., Philipps, V., and Liquet, B. (2017).
\newblock Estimation of extended mixed models using latent classes and latent processes: The r package lcmm.
\newblock {\em Journal of Statistical Software, Articles} {\bf 78,} 1--56.

\bibitem[\protect\citeauthoryear{Proust-Lima, S{\'e}ne, Taylor, and Jacqmin-Gadda}{Proust-Lima et~al.}{2014}]{proust2014joint}
Proust-Lima, C., S{\'e}ne, M., Taylor, J.~M., and Jacqmin-Gadda, H. (2014).
\newblock Joint latent class models for longitudinal and time-to-event data: A review.
\newblock {\em Statistical methods in medical research} {\bf 23,} 74--90.

\bibitem[\protect\citeauthoryear{Raith, Udy, Bailey, McGloughlin, MacIsaac, Bellomo, Pilcher, et~al\mbox{.}}{Raith et~al.}{2017}]{raith2017prognostic}
Raith, E.~P., Udy, A.~A., Bailey, M., McGloughlin, S., MacIsaac, C., Bellomo, R., Pilcher, D.~V., et~al. (2017).
\newblock Prognostic accuracy of the sofa score, sirs criteria, and qsofa score for in-hospital mortality among adults with suspected infection admitted to the intensive care unit.
\newblock {\em Jama} {\bf 317,} 290--300.

\bibitem[\protect\citeauthoryear{Rizopoulos}{Rizopoulos}{2016}]{2017_JMBayes}
Rizopoulos, D. (2016).
\newblock The r package jmbayes for fitting joint models for longitudinal and time-to-event data using mcmc.
\newblock {\em Journal of Statistical Software, Articles} {\bf 72,} 1--46.

\bibitem[\protect\citeauthoryear{Rizopoulos and Ghosh}{Rizopoulos and Ghosh}{2011}]{rizopoulos2011bayesian}
Rizopoulos, D. and Ghosh, P. (2011).
\newblock A bayesian semiparametric multivariate joint model for multiple longitudinal outcomes and a time-to-event.
\newblock {\em Statistics in medicine} {\bf 30,} 1366--1380.

\bibitem[\protect\citeauthoryear{Rustand, Van~Niekerk, Krainski, Rue, and Proust-Lima}{Rustand et~al.}{2024}]{rustand2024fast}
Rustand, D., Van~Niekerk, J., Krainski, E.~T., Rue, H., and Proust-Lima, C. (2024).
\newblock Fast and flexible inference for joint models of multivariate longitudinal and survival data using integrated nested laplace approximations.
\newblock {\em Biostatistics} {\bf 25,} 429--448.

\bibitem[\protect\citeauthoryear{Simon, Friedman, Hastie, and Tibshirani}{Simon et~al.}{2013}]{simon2013sparse}
Simon, N., Friedman, J., Hastie, T., and Tibshirani, R. (2013).
\newblock A sparse-group lasso.
\newblock {\em Journal of Computational and Graphical Statistics} {\bf 22,} 231--245.

\bibitem[\protect\citeauthoryear{Tsiatis and Davidian}{Tsiatis and Davidian}{2004}]{tsiatis2004joint}
Tsiatis, A.~A. and Davidian, M. (2004).
\newblock Joint modeling of longitudinal and time-to-event data: an overview.
\newblock {\em Statistica Sinica} pages 809--834.

\bibitem[\protect\citeauthoryear{Uno, Cai, Pencina, D'Agostino, and Wei}{Uno et~al.}{2011}]{uno2011c}
Uno, H., Cai, T., Pencina, M.~J., D'Agostino, R.~B., and Wei, L. (2011).
\newblock On the c-statistics for evaluating overall adequacy of risk prediction procedures with censored survival data.
\newblock {\em Statistics in medicine} {\bf 30,} 1105--1117.

\bibitem[\protect\citeauthoryear{Virtanen, Gommers, Oliphant, Haberland, Reddy, Cournapeau, Burovski, Peterson, Weckesser, Bright, et~al\mbox{.}}{Virtanen et~al.}{2020}]{virtanen2020scipy}
Virtanen, P., Gommers, R., Oliphant, T.~E., Haberland, M., Reddy, T., Cournapeau, D., Burovski, E., Peterson, P., Weckesser, W., Bright, J., et~al. (2020).
\newblock Scipy 1.0: fundamental algorithms for scientific computing in python.
\newblock {\em Nature methods} {\bf 17,} 261--272.

\bibitem[\protect\citeauthoryear{Wang, Shen, and Boye}{Wang et~al.}{2012}]{wang2012joint}
Wang, P., Shen, W., and Boye, M.~E. (2012).
\newblock Joint modeling of longitudinal outcomes and survival using latent growth modeling approach in a mesothelioma trial.
\newblock {\em Health Services and Outcomes Research Methodology} {\bf 12,} 182--199.

\bibitem[\protect\citeauthoryear{Wulfsohn and Tsiatis}{Wulfsohn and Tsiatis}{1997}]{wulfsohn1997joint}
Wulfsohn, M.~S. and Tsiatis, A.~A. (1997).
\newblock A joint model for survival and longitudinal data measured with error.
\newblock {\em Biometrics} pages 330--339.

\bibitem[\protect\citeauthoryear{Yu, Law, Taylor, and Sandler}{Yu et~al.}{2004}]{yu2004joint}
Yu, M., Law, N.~J., Taylor, J.~M., and Sandler, H.~M. (2004).
\newblock Joint longitudinal-survival-cure models and their application to prostate cancer.
\newblock {\em Statistica Sinica} pages 835--862.

\bibitem[\protect\citeauthoryear{Yuan, Liu, and Ye}{Yuan et~al.}{2011}]{yuan2011efficient}
Yuan, L., Liu, J., and Ye, J. (2011).
\newblock Efficient methods for overlapping group lasso.
\newblock {\em Advances in neural information processing systems} {\bf 24,} 352--360.

\bibitem[\protect\citeauthoryear{Zhu, Byrd, Lu, and Nocedal}{Zhu et~al.}{1997}]{zhu1997algorithm}
Zhu, C., Byrd, R.~H., Lu, P., and Nocedal, J. (1997).
\newblock Algorithm 778: L-bfgs-b: Fortran subroutines for large-scale bound-constrained optimization.
\newblock {\em ACM Transactions on Mathematical Software (TOMS)} {\bf 23,} 550--560.

\bibitem[\protect\citeauthoryear{Zou and Hastie}{Zou and Hastie}{2005}]{zou2005regularization}
Zou, H. and Hastie, T. (2005).
\newblock Regularization and variable selection via the elastic net.
\newblock {\em Journal of the Royal Statistical Society: Series B (Statistical Methodology)} {\bf 67,} 301--320.

\end{thebibliography}

\clearpage
\begin{appendices}
\setcounter{equation}{0}
\begin{center}
    \Large \textbf{Supplementary Material} 
\end{center}
\bigskip



\label{sec:sum_notation}
To help the reader, Table \ref{tab:notation} provides a list of notations used in the the paper.
\begin{table}[h]
\begin{center}
\caption{Summary of notation}
\label{tab:notation}
\begin{tabular}{cl } 
 \toprule
 \textbf{Notation} & \textbf{Definition} \\ 
 \midrule
 $X_i$ & Time-independent feature \\
 $Y_i$ & Longitudinal markers \\ 
 $T_i$ & Survival time \\ 
 $\Delta_i$ & Censoring indication \\ 
 $g_i$ & Latent class membership variable\\ 
 $b_i$ & Random effects of the longitudinal markers \\
 $n$ & Number of training samples \\
 $p$ & Number of time-independent features \\
 $L$ & Number of longitudinal markers \\
 $K$ & Number of latent classes\\
 $D$ & Variance-covariance matrix of the $b_i$\\
 $U_i$ & Fixed-effect design matrix \\ 
 $V_i$ & Random-effect design matrix \\ 
 $\psi_i$ & Association features \\ 
 $\xi$ & Time-independent parameters \\
 $\beta$ & Fixed-effect parameter\\
 $\gamma$ & Joint association parameter \\ 
 $\lambda_0$ & Baseline hazard function  \\ 
 $I_m$ & Identity matrix of $\R^{m \times m}$ \\
 $\mathbf{1}_m$ & Vector of $\R^m$ having all coordinates equal to one \\
 $\mathbf{0}_m$ & Vector of $\R^m$ having all coordinates equal to zero \\
 $\norm{\cdot}_q$ & The usual $\ell_q$-quasi norm, $q > 0$\\
 \bottomrule
\end{tabular}
\end{center}
\end{table}

\section{Details on the extended EM Algorithm}
\label{sec:ext-EM_sup}
We detail in this section our learning methodology. First, recall that the penalized negative log-likelihood is defined by
\begin{equation}
  \label{eq:pen-log-lik_ext}
 \cL_n^\text{pen}(\theta) =  \cL_n(\theta) + \sum_{k=1}^{K} \zeta_{1,k} \Omega_1(\xi_k) + \sum_{k=1}^{K} \zeta_{2,k} \Omega_2(\gamma_k),
\end{equation}
where
\[\Omega_1(\xi_k) =  (1-\eta)\norm{\xi_k}_1 + \dfrac\eta2 \norm{\xi_k}_2^2 \quad \text{and} \quad \Omega_2(\gamma_k) = (1-\tilde{\eta})\norm{\gamma_k}_1 + \tilde{\eta} \sum_{\ell=1}^L \norm{\gamma_k^\ell}_2, \]
where $(\eta, \tilde{\eta}) \in [0, 1]^2$ are fixed (depending on the level of sparsity expected), $\gamma_k^\ell = (\gamma_{k,1}^\ell, \ldots, \gamma_{k,M}^\ell )^\top \in \R^{M}$ is the subset of $\gamma_k$ corresponding to the longitudinal marker $\ell$, $\norm{\cdot}_{1}$ (resp. $\norm{\cdot}_{2}$) denotes the usual $\ell_{1}$ (resp. $\ell_{2}$) norm. In all our experiments, we take $\eta=0.1$ and $\tilde{\eta}=0.9$.

The goal is to minimize this objective function by the EM algorithm.
This is done in two steps: compute the expectation of the negative complete log-likelihood with respect to the unobserved quantities (the random effects $b_i$ and the latent classes $g_i$), then minimize the obtained quantity with respect to all parameters of the model in $\theta$. For simplicity, we won't compute all terms of the expectation in the E-step but only the quantities used in the M-step. Moreover, we perform minimization with respect to $\theta$ in several steps, minimizing with respect to block of parameters separately to obtain tractable updates. 


\subsection{E-step}
\label{sec:E-step}

Recall that, under our assumptions, the negative complete log-likelihood writes
\begin{align}
\label{eq:comp-log-lik_ext}
\cL_n^\text{comp}(\theta) = - n^{-1} \sum_{i=1}^n \Big( \log f_\theta(b_i) +  \sum_{k=1}^{K} \ind{\{g_i=k\}} \big( & \log \P_\theta(g_i=k) + \log f_\theta(Y_i \, | \, b_i, g_i = k)\big) \\
  & + \log f_\theta(T_i, \Delta_i \, | \, Y_i, g_i = k) \Big).
\end{align}
Let us introduce a few matrix notations. Concatenating all longitudinal markers and all observation times, the mean of the vector $Y_i \, | \, b_i, g_i = k$ (defined in \eqref{eq:Y_i_density} in the main paper) can be rewritten ${M_{ik} = U_i\beta_k + V_ib_i}$, where we introduce the design matrices
\[ U_i = 
\begin{bmatrix}
  U_i^1 & \cdots & 0\\
  \vdots &  \ddots & \vdots \\
  0 & \cdots & U_i^L
\end{bmatrix} 
\in \R^{n_i \times q}
\qquad \text{and} \qquad
V_i = 
\begin{bmatrix}
  V_i^1 & \cdots & 0\\
  \vdots &  \ddots & \vdots \\
  0 & \cdots & V_i^L
\end{bmatrix}
\in \R^{n_i \times r}
\]
and for all $\ell \in \{1, \ldots, L\}$, one writes
\[ U_i^\ell = 
\begin{bmatrix}
  {u^\ell(t_{i1}^\ell)}^\top\\ \vdots \\ {u^\ell(t_{in_i^\ell}^\ell)^\top }
\end{bmatrix}
\in \R^{n_i^\ell \times q_\ell}
\qquad \text{and} \qquad
V_i^\ell = 
\begin{bmatrix}
  {v^\ell(t_{i1}^\ell)^\top}\\ \vdots \\ {v^\ell(t_{in_i^\ell}^\ell)^\top}
\end{bmatrix}
\in \R^{n_i^\ell \times r_\ell}.
\]
Under all assumptions of Section \ref{sec:method} in the main paper, we can then write explicitly the different terms. The random effects simply follow a Gaussian distribution, which yields
\begin{align*}
    \log f_\theta(b_i) &= -\dfrac12 (r \log 2\pi + \log \det(D) + b_i^\top D^{-1}b_i).
\end{align*}
Then, the conditional density of the longitudinal features (in the Gaussian case) writes
\begin{equation*}
    \log f_\theta(Y_i \, | \, b_i, g_i = k) 
     =  -\dfrac12 \big(n_i \log 2\pi + \log \det(\Sigma_i) + (Y_i - U_i\beta_k - V_ib_i)^\top \Sigma_i^{-1}(Y_i - U_i\beta_k - V_ib_i)\big).
\end{equation*}
Finally, the survival terms in the complete likelihood write
\begin{align}
\label{eq:density_T_i_Delta_i}
     \log f_\theta(T_i, \Delta_i \, | \, Y_i, g_i = k) = \Delta_i \big(\log \lambda_0(T_i) + \psi_i(T_i)^\top \gamma_k \big) - \sum_{j=1}^{J} \lambda_0(\tau_j) \exp \big( \psi_i(\tau_j)^\top \gamma_k \big) \ind{\{\tau_j \leq T_i\}}.
\end{align}
We can therefore decompose the expected negative complete log-likelihood as
\begin{align*}
\cQ_n(\theta, \theta^{(w)}) &= \E_{\theta^{(w)}}[\cL_n^\text{comp}(\theta) \,| \, \cD_n] \\
 &= - n^{-1}\sum_{i=1}^{n}\Big( A^{1}_i(D) + \sum_{k=1}^K \P_{\theta^{(w)}}(g_i=k | \mathcal{D}_n) \big(A^{2}_{i,k}(\xi) + A^{3}_i(\beta_k, \Sigma_i) + A^{4}_i(\gamma_k, \lambda_0) \big) \Big) \\
 & \quad + \text{constants},
\end{align*}
where, for any $b \in \R^{r}$, $D \in \R^{r \times r}$, $\beta \in \R^{q}$, $\Sigma \in \R^{n_i}$, we define
\begin{align*}
    A^1_i(D) &= \frac{1}{2}\E_{\theta^{(w)}} \big[\log \det(D) + b_i^\top D^{-1}b_i \, | \, \mathcal{D}_n \big], \\
    A^{2}_{i,k}(\xi) &= \log \Big(\dfrac{e^{X_i^\top\xi_k}}{\sum_{j=1}^{K} e^{X_i^\top\xi_j}} \Big),\\
    A^{3}_i(\beta, \Sigma) &= \frac{1}{2} \E_{\theta^{(w)}} \big[\log \det(\Sigma) + (Y_i - U_i\beta - V_ib_i)^\top \Sigma^{-1}(Y_i - U_i\beta - V_ib_i) \, | \, \mathcal{D}_n\big],\\
    A^{4}_i(\gamma, \lambda_0) &= \Delta_i \big(\log \lambda_0(T_i) + \psi_i(T_i)^\top \gamma \big) - \sum_{j=1}^{J} \lambda_0(\tau_j) \exp \big(  \psi_i(\tau_j)^\top \gamma \big) \ind{\{\tau_j \leq T_i\}}.
\end{align*}
As explained before, we do not compute the expectation of the negative complete likelihood but only of the quantities needed in the M-step. First, the following technical lemmas will prove useful in the E-step and makes use of conjugate properties of Gaussian distributions.

\begin{lemma}\label{lemma:dist_Y_b}
For any $i \in \{1, \dots, n \}$, $k \in \{1, \dots K \}$,
\begin{equation*}
    Y_i \, | \, g_i = k \sim \cN\big(U_i\beta_k, V_iDV_i^\top + \Sigma_i\big)  \quad \text{and} \quad  b_i \, | \,Y_i, g_i=k \sim \cN \big(O_{i,k}, W_i\big),
\end{equation*}
where
\begin{equation}
    \label{eq:def_O_W}
    O_{i,k}=W_iV_i^T{\Sigma_i}^{-1}(Y_i - U_i\beta_k) \quad \text{and} \quad W_i = \big(V_i^\top {\Sigma_i}^{-1}V_i + {D}^{-1} \big)^{-1}.
\end{equation}
\end{lemma}
\begin{proof}
From \eqref{eq:Y_i_density} in the main paper we know that $Y_i \, | \, b_i, g_i = k \sim \cN\big(M_{ik},\Sigma_i\big)$ and $b_i \sim \cN\big(0, D) $, which gives
\begin{equation*}
    Y_i | g_i = k \sim \cN\big(U_i\beta_k, V_iDV_i^\top + \Sigma_i\big).
\end{equation*}
Moreover, by Bayes's rule, the distribution of $b_i | Y_i, g_i=k$ can be written
\begin{align*}
f_\theta(b_i | Y_i, g_i=k) &\propto f_\theta(Y_i | b_i, g_i=k) f_\theta(b_i | g_i=k)\nonumber\\
&\propto \exp{\big((Y_i - U_i\beta_k - V_ib_i)^\top \Sigma_i^{-1}(Y_i - U_i\beta_k - V_ib_i) + b_i^\top D^{-1}b_i}\big)\nonumber\\
 &\propto \exp{\big((b_i - O_{i,k})^\top W_i^{-1}(b_i - O_{i,k})\big)},
\end{align*}
where $O_{i,k}=W_iV_i^T{\Sigma_i}^{-1}(Y_i - U_i\beta_k)$ and $W_i = \big(V_i^\top \Sigma_i^{-1}V_i + D^{-1} \big)^{-1}$. We then have
\begin{equation*}
  b_i | Y_i, g_i=k \sim \cN \big(O_{i,k}, W_i\big).
\end{equation*}
\end{proof}

The following lemma gives the three expectations that appear in this E-step.

\begin{lemma}
\label{lemma:E_mu_b_i_gi}
    For any $i \in \{1, \dots, n \}$, $k \in \{1, \dots, K \}$, $\theta \in \R^P$, the three following integrals are closed-form and write
    \begin{equation}
    \label{eq:E_b}
    \E_{\theta}[ b_i | \mathcal{D}_n] = \dfrac{ \sum_{j=1}^{K} \P_\theta(g_i=j) f_\theta(T_i, \Delta_i | Y_i, g_i=j) f_{\theta}(Y_i | g_i = j) O_{i,j}}{\sum_{j=1}^{K} \P_\theta(g_i=j) f_\theta(T_i, \Delta_i | Y_i, g_i=j) f_{\theta}(Y_i | g_i = j)},
    \end{equation}
    
    \begin{equation}
    \label{eq:E_bbT}
     \E_{\theta}[ b_ib_i^\top | \mathcal{D}_n] = \dfrac{ \sum_{j=1}^{K} \P_\theta(g_i=j) f_\theta(T_i, \Delta_i | Y_i, g_i=j) f_{\theta}(Y_i | g_i = j) (W_i + O_{i,j}{O_{i,j}}^\top)}{\sum_{j=1}^{K} \P_\theta(g_i=j) f_\theta(T_i, \Delta_i | Y_i, g_i=j) f_{\theta}(Y_i | g_i = j)},
     \end{equation}
     and
     \begin{equation}
     \label{eq:E_ind}
      \tilde{\pi}^{\theta}_{ik} = \P_\theta(g_i=k | \mathcal{D}_n) = \dfrac{\P_\theta(g_i=k) f_\theta(T_i, \Delta_i | Y_i, g_i=k) f_{\theta}(Y_i | g_i = k)}{\sum_{j=1}^{K} \P_\theta(g_i=j) f_\theta(T_i, \Delta_i | Y_i, g_i=j) f_{\theta}(Y_i | g_i = j)},
     \end{equation}
    where  
    $f_\theta(Y_i | g_i = j)$ is the density of the multivariate Gaussian distribution of Lemma \ref{lemma:dist_Y_b}.
\end{lemma}
\begin{proof}
By Assumption~\ref{indep-hyp-2} in the main paper, the distribution of $b_i$ given the observed data $\mathcal{D}_n$, for any $\theta \in \R^P$, writes
\begin{align*}
    f_{\theta}(b_i | \mathcal{D}_n)
    &= \dfrac{f_{\theta}(b_i, T_i, \Delta_i, Y_i)}{f_{\theta}(T_i, \Delta_i, Y_i)} \\
    &= \frac{1}{f_{\theta}(T_i, \Delta_i, Y_i)} \sum_{j=1}^{K} \P_\theta(g_i=j)f_{\theta}(b_i, T_i, \Delta_i, Y_i | g_i=j)\\
    &= \frac{1}{f_{\theta}(T_i, \Delta_i, Y_i)} \sum_{j=1}^{K} \P_\theta(g_i=j) f_{\theta}(Y_i | g_i=j) f_{\theta}(T_i, \Delta_i| Y_i, g_i=j) f_{\theta}(b_i|Y_i, g_i=j).
\end{align*}
Similarly, we have
\begin{equation*}
    f_{\theta}(T_i, \Delta_i, Y_i) = \sum_{j=1}^{K}\P_\theta(g_i=j) f_{\theta}(Y_i | g_i=j) f_{\theta}(T_i, \Delta_i| Y_i, g_i=j).
\end{equation*}
This gives that, for any function $\mu$, 
\begin{align*}
  \E_{\theta}[ \mu(b_i) | \mathcal{D}_n] &= \int_{\R^r} \mu(b_i) f_{\theta}(b_i | \mathcal{D}_n) \dd b_i \\
& = \frac{1}{f_{\theta}(T_i, \Delta_i, Y_i)}\sum_{j=1}^{K} \P_\theta(g_i=j)f_{\theta}(Y_i | g_i=j) f_{\theta}(T_i, \Delta_i| Y_i, g_i=j) \int_{\R^r} \mu(b_i) f_{\theta}(b_i|Y_i, g_i=j) \dd b_i\\
   & = \frac{1}{f_{\theta}(T_i, \Delta_i, Y_i)}\sum_{j=1}^{K} \P_\theta(g_i=j)f_{\theta}(Y_i | g_i=j) f_{\theta}(T_i, \Delta_i| Y_i, g_i=j) \E_{\theta}[\mu(b_i) |Y_i, g_i=j]
\end{align*}
Hence, for the functions $\mu(b_i) = b_i$ and $\mu(b_i) = b_i b_i^\top$, we obtain the result by applying Lemma \ref{lemma:dist_Y_b}, which gives
\[\E_{\theta}[ b_i | Y_i, g_i = j] = O_{i,j}, \quad \text{and} \quad \E_{\theta}[ b_ib_i^\top | Y_i, g_i = j] = W_i + O_{i,j}{O_{i,j}}^\top.\]
In the same manner, we can see that for any $k \in \{1, \dots, K \}$,
\begin{align*}
   \E_{\theta}[\ind{\{g_i=k\}} | T_i, \Delta_i, Y_i] &= \P_{\theta}(g_i=k | T_i, \Delta_i, Y_i) \\
   &= \frac{1}{f_\theta(T_i, \Delta_i, Y_i)} \P_\theta(g_i=k) f_\theta(T_i, \Delta_i | Y_i, g_i=k) f_\theta(Y_i | g_i=k),
\end{align*}
which concludes the proof.
\end{proof}
\paragraph{Extensions of the longitudinal model} To support other types of longitudinal features (e.g., count, binary, etc) where the conditional longitudinal features $Y_i \, | \, b_i, g_i = k$ are not Gaussian, the simple solution is to choose distributions such that  the random effect $b_i$ and $Y_i \, | \, b_i, g_i = k$ have conjugate distributions. There are several options, for example using the beta-binomial distribution for binary features or negative-binomial for count features. In this case, the expectation $\E_{\theta}[\mu(b_i) |Y_i, g_i=j]$ and the density $f_{\theta}(b_i|Y_i, g_i=j)$ can be computed in closed forms~\citep{molenberghs2010family} and the extension is trivial. If the distributions are not conjugate, some numerical integration methods can be used to compute $\E_{\theta}[\mu(b_i) |Y_i, g_i=j]$ and $f_{\theta}(b_i|Y_i, g_i=j)$~\citep[see, e.g.,][]{fabio2012poisson}.

\subsection{M-step: closed-form updates}
\label{sec:step-closed-form}
Now, we assume that we are at step $w+1$ of the algorithm, meaning that we have a current value $\theta^{(w)}$ of the parameters and we update it to get the new parameters $\theta^{(w+1)}$ by solving
\begin{equation}
     \theta^{(w+1)} \in \underset{\theta \in \R^P}{\argmin} \, \cQ^\text{pen}_n(\theta, \theta^{(w)}).
\end{equation}
 Note that we update for the coordinates of $\theta^{(w)}$ in order, which is $(D^{(w+1)},
(\xi_k^{(w+1)})_{k \in \{1, \dots, K\}},
(\beta_k^{(w+1)})_{k \in \{1, \dots, K\}},\\
(\gamma_k^{(w+1)})_{k \in \{1, \dots, K\}},
\lambda_0^{(w+1)}$, $\phi^{(w+1)})$. Then, following this order, the update of later coordinates use the latest update of the previous ones. The update of several coordinates of $\theta^{(w+1)}$ can be obtained in closed-form. 

\begin{lemma}
At step $w + 1$ of the EM algorithm, the update of $D$ is
\begin{equation*}
  \label{eq:update_D}
  D^{(w+1)} = n^{-1} \sum_{i=1}^n \E_{\theta^{(w)}}[ b_i b_i^\top | T_i, \Delta_i, Y_i],
\end{equation*}
where $\E_{\theta^{(w)}}[ b_i b_i^\top | T_i, \Delta_i, Y_i]$ is given by \eqref{eq:E_bbT}.
\end{lemma}

\begin{proof}
    The update for $D^{(w)}$ requires to solve the following minimization problem
\begin{equation}
  \label{eq:minimization-beta_k}
  D^{(w+1)} \in  \underset{D \in \R^{q \times q}}{\argmin} \,  -n^{-1} \sum_{i=1}^{n} A^1_i(D),
\end{equation}
We have for any $i \in \{1, \dots, n \}$,
\begin{align*}
   A^1_i(D) &= \E_{\theta^{(w)}}\big[\log \det(D) + b_i^{\top} D^{-1} b_i \big]  = \log \det(D) + \int_{\R^r} b_i^\top D^{-1}b_i f_{\theta^{(w)}}(b_i | T_i, \Delta_i, Y_i) \dd b_i.
\end{align*}
The gradient of $A^1_i$ is here given by
\begin{align*}
  \frac{\partial A^1_i(D)}{\partial D} &= D^{-\top} - \int_{\R^r} D^{-\top} b_i b_i^\top D^{-\top} f_{\theta^{(w)}}(b_i | T_i, \Delta_i, Y_i) \dd b_i  \\
  &= D^{-\top} -  D^{-\top} \E_{\theta^{(w)}}[ b_i b_i^\top | T_i, \Delta_i, Y_i] D^{-\top},
\end{align*}
where $D^{-\top}$ is the transpose of matrix $D^{-1}$. 
The proof is completed by cancelling the gradient, that is 
\begin{align*}
    -n^{-1} \sum_{i=1}^{n} \frac{\partial A^1_i(D)}{\partial  D} = 0
 \quad \Leftrightarrow \quad - D^{-\top} + n^{-1} \sum_{i=1}^{n}D^{-\top} \E_{\theta^{(w)}}[ b_i b_i^\top | T_i, \Delta_i, Y_i] D^{-\top} = 0
\end{align*}
\end{proof}

\begin{lemma}
At step $w + 1$ of the EM algorithm, the update of $\beta_k$ is
\begin{equation}
  \label{eq:update_beta}
  {\beta_k}^{(w+1)}= \Big(\sum_{i=1}^n \tilde{\pi}_{ik}^{\theta^{(w)}} U_i^\top U_i\Big)^{-1}\Big(\sum_{i=1}^n \tilde{\pi}_{ik}^{\theta^{(w)}} U_i^\top (Y_i - V_i \E_{\theta^{(w)}}[ b_i | T_i, \Delta_i, Y_i])\Big),
\end{equation}
where $\E_{\theta^{(w)}}[ b_i | T_i, \Delta_i, Y_i]$ and $\tilde{\pi}_{ik}^{\theta^{(w)}}$ are given respectively by \eqref{eq:E_b} and \eqref{eq:E_ind} .
\end{lemma}
\begin{proof}
    The update for $\beta_k^{(w)}$ requires to solve the following minimization problem
\begin{equation}
  \beta_k^{(w+1)} \in  \underset{\beta \in \R^q}{\argmin} \, -n^{-1}\sum_{i=1}^n  \tilde{\pi}_{ik}^{\theta^{(w)}} A_i^{3}\big(\beta, \Sigma_i^{(w)} \big).
\end{equation}
We have
\begin{align*}
 A_i^{3}(\beta, \Sigma_i^{(w)})
  & = \frac{1}{2}\E_{\theta^{(w)}} \big[\log \det(\Sigma_i^{(w)}) + (Y_i - U_i\beta - V_ib_i)^\top {\Sigma_i^{(w)}}^{-1}(Y_i - U_i\beta - V_ib_i) \, | \, \mathcal{D}_n\big]\\
& = \frac{1}{2}\log \det(\Sigma_i^{(w)}) + \frac{1}{2} \E_{\theta^{(w)}} \big[ (Y_i - U_i\beta - V_ib_i)^\top {\Sigma_i^{(w)}}^{-1}(Y_i - U_i\beta - V_ib_i) \, | \, \mathcal{D}_n\big]\\
& = \E_{\theta^{(w)}} \big[ - (Y_i - V_ib_i)^\top {\Sigma_i^{(w)}}^{-1} U_i\beta  +  \frac{1}{2} \beta^\top U_i^\top {\Sigma_i^{(w)}}^{-1} U_i \beta \, | \, \mathcal{D}_n\big] + \text{constants}
\\
& =  - (Y_i - V_i \E_{\theta^{(w)}} \big[b_i| \, \mathcal{D}_n\big])^\top {\Sigma_i^{(w)}}^{-1} U_i\beta  +  \frac{1}{2} \beta^\top U_i^\top {\Sigma_i^{(w)}}^{-1} U_i \beta \,  + \text{constants}
\end{align*}
(where we treat as constants all quantities independent of $\beta$). Then, the gradient of $A_i^{3}$ writes
\begin{align*}
  \frac{\partial A_i^{3}(\beta, \Sigma_i^{(w)})}{\partial \beta}
  &= - (Y_i - V_i \E_{\theta^{(w)}}[b_i | \mathcal{D}_n])^\top {\Sigma_i^{(w)}}^{-1} U_i +\beta^\top U_i^\top {\Sigma_i^{(w)}}^{-1}U_i.
\end{align*}
Given the form of $U_i$ and $V_i$ along with the fact that $\Sigma_i^{(w)}$ is a diagonal matrix, we can rewrite the gradient of $A_i^{3}$ as
\begin{equation*}
  \frac{\partial A_i^{3}(\beta, \Sigma_i^{(w)})}{\partial \beta} = - (Y_i - V_i \E_{\theta^{(w)}}[b_i | \mathcal{D}_n])^\top U_i {\tilde{\Sigma}^{(w)}}^{-1} +\beta^\top U_i^\top U_i {{\tilde\Sigma}^{(w)}}^{-1},
\end{equation*}
where $\tilde\Sigma^{(w)}$ is the diagonal matrix whose diagonal is ${({\phi_1^{(w)}{\textbf{1}_{q_1}}}^\top, \ldots, {\phi_L^{(w)} {\textbf{1}_{q_L}}}^\top)}^\top \in \R^{q}$. The closed-form update of $\beta_k^{(w)}$ is then obtained by canceling the gradient, that is
\begin{align*}
    & -n^{-1}\sum_{i=1}^n \tilde{\pi}_{ik}^{\theta^{(w)}} \frac{\partial A_i^{3}(\beta, \Sigma_i^{(w)})}{\partial \beta} = 0\\
    \Leftrightarrow \quad & \sum_{i=1}^n \tilde{\pi}_{ik}^{\theta^{(w)}} \Big((Y_i - V_i \E_{\theta^{(w)}}[b_i | \mathcal{D}_n])^\top U_i {\tilde{\Sigma}^{(w)}}^{-1} - \beta^\top U_i^\top U_i {{\tilde\Sigma}^{(w)}}^{-1}\Big) = 0 \\
    \Leftrightarrow \quad & \sum_{i=1}^n \tilde{\pi}_{ik}^{\theta^{(w)}} \Big((Y_i - V_i \E_{\theta^{(w)}}[b_i | \mathcal{D}_n])^\top U_i - \beta^\top U_i^\top U_i \Big) = 0 \\
    \Leftrightarrow \quad & \beta = \Big(\sum_{i=1}^n \tilde{\pi}_{ik}^{\theta^{(w)}} U_i^\top U_i\Big)^{-1}\Big(\sum_{i=1}^n \tilde{\pi}_{ik}^{\theta^{(w)}} U_i^\top (Y_i - V_i \E_{\theta^{(w)}}[b_i | \mathcal{D}_n])\Big).
\end{align*}
\end{proof}

\begin{lemma}
At step $w + 1$ of the EM algorithm, for any $j \in \{1, \dots, J\}$, given the the update of $\lambda_0(\tau_j)$ is
\begin{equation}
  \label{eq:update_lambda_0}
  \lambda_0^{(w+1)}(\tau_j)= \dfrac{\sum_{i=1}^n \Delta_i \ind{\{T_i=\tau_j\}}}{\sum_{i=1}^n \sum_{k=1}^{K} \tilde\pi_{ik}^{\theta^{(w)}} \exp \big(\psi_i(\tau_j)^\top \gamma_k^{(w+1)}\big) \ind{\{T_i \geq \tau_j\}} },
\end{equation}
where $\tilde\pi_{ik}^{\theta^{(w)}} $ is given by \eqref{eq:E_ind}.
\end{lemma}
\begin{proof}
The update for $\lambda_0^{(w)}(\tau_j)$ requires to solve the minimization problem
\begin{equation}
  \label{eq:eq:minimization-baseline-hazard}
  \lambda_0^{(w+1)}(\tau_j) \in  \underset{\lambda_0 \in \R}{\argmin} \, -n^{-1} \sum_{i=1}^{n} \sum_{k=1}^{K} \tilde\pi_{ik}^{\theta^{(w)}} A^{4}_i(\gamma_k^{(w+1)}, \lambda_0).
\end{equation}
We have 
\begin{align*}
     &-n^{-1} \sum_{i=1}^{n} \sum_{k=1}^{K} \tilde\pi_{ik}^{\theta^{(w)}} A^{4}_i(\gamma_k^{(w+1)}, \lambda_0) \\
     & \quad = -n^{-1} \sum_{i=1}^{n} \sum_{k=1}^{K} \tilde\pi_{ik}^{\theta^{(w)}} \Big( \Delta_i \big(\log \lambda_0(T_i) + \psi_i(T_i)^\top \gamma_k^{(w+1)} \big) - \sum_{j=1}^{J} \lambda_0(\tau_j) \exp \big(  \psi_i(\tau_j)^\top \gamma_k^{(w+1)} \big) \ind{\{\tau_j \leq T_i\}} \Big) \\
    & \quad = -n^{-1} \sum_{i=1}^{n} \Delta_i \log \lambda_0(T_i) + n^{-1} \sum_{i=1}^{n} \sum_{k=1}^{K} \tilde\pi_{ik}^{\theta^{(w)}} \sum_{j=1}^{J} \lambda_0(\tau_j) \exp \big(  \psi_i(\tau_j)^\top \gamma_k^{(w+1)} \big) \ind{\{\tau_j \leq T_i\}} + \text{constants}
\end{align*}
(where we keep only the terms with $\tau_j$ for any $j \in \{1, \dots, J \}$). By taking the gradient of the previous expression over $\lambda_0(\tau_j)$ and setting it to zero, that is
\begin{equation}
  -\frac{n^{-1}}{\lambda_0}\sum_{i=1}^n  \Delta_i \ind{\{T_i=\tau_j\}}  + n^{-1} \sum_{i=1}^n\sum_{k=1}^{K} \tilde\pi_{ik}^{\theta^{(w)}} \exp \big(  \psi_i(\tau_j)^\top \gamma_k^{(w+1)}\big) \ind{\{T_i \geq \tau_j\}} = 0.
\end{equation}
We then obtain the update for $\lambda_0^{(w)}(\tau_j)$ as the desired result.
\end{proof}

Note that the closed-form update of $\lambda_0^{(w)}(\tau_j)$ is a Breslow-like estimator~\citep{breslow1972contribution} adapted to our model. Finally, recall that $\Sigma_i$ is a diagonal matrix with a diagonal of the form  $(\phi_1{\textbf{1}_{n_i^1}}^\top \cdots \phi_L {\textbf{1}_{n_i^L}}^\top)$. Estimating $\Sigma_i$ thus amounts to estimating $\phi_1, \dots, \phi_L$, whose updates are given in the following lemma.
\begin{lemma}
At step $w + 1$ of the EM algorithm, the update of $\phi_\ell$ is
\begin{align}
  \label{eq:update_phi}
  \phi_\ell^{(w+1)}= \frac{1}{\sum_{i=1}^n n_i^\ell } \sum_{i=1}^n \sum_{k=1}^{K} \tilde\pi_{ik}^{\theta^{(w)}} \Big( &\big(Y_i^\ell - U_i^\ell \beta_k^{\ell (w+1)}\big)^\top \big(Y_i^\ell - U_i^\ell \beta_k^{\ell (w+1)} - 2V_i^\ell \E_{\theta^{(w)}}[ b_i^\ell \,| \, \mathcal{D}_n]\big) \nonumber\\
  & + \Tr\big({V_i^\ell}^\top V_i^\ell \E_{\theta^{(w)}}[ b_i^\ell b_i^{\ell\top} \, | \, \mathcal{D}_n]\big)\Big),
\end{align}
where $\E_{\theta^{(w)}}[ b_i^\ell\, | \, \mathcal{D}_n]$ and 
$\E_{\theta^{(w)}}[ b_i^\ell b_i^{\ell\top} \, | \, \mathcal{D}_n]$ are obtained from \eqref{eq:E_b} and \eqref{eq:E_bbT}. 
\end{lemma}
\begin{proof}
The update for $\phi^{(w)}$ requires to solve the following minimization problem
\begin{equation}
  \label{eq:eq:minimization-phi}
  \phi^{(w+1)} \in  \underset{\phi \in \R^L}{\argmin}  \, -n^{-1} \sum_{i=1}^n \sum_{k=1}^{K} \tilde\pi_{ik}^{\theta^{(w)}} A^{3}_i(\beta_k^{(w+1)}, \text{Diag}_i(\phi)),
\end{equation}
where we denote by $\text{Diag}_i(\phi)$ the diagonal matrix $\Sigma_i$ to make clear its dependence on $\phi$. We let
${M_{ik}}^{(w+1)} = U_i{\beta_k}^{(w+1)} + V_ib_i$, $M_{ik}^{\ell (w+1)} = U_i^\ell \beta_k^{\ell(w+1)} + V_i^\ell b_i^\ell$. Then, taking advantage of the structure of $\text{Diag}_i(\phi)$, we have
\begin{align*}
    A^{3}_i(\beta_k^{(w+1)}, \text{Diag}_i(\phi)) & =\frac{1}{2} \log \det(\text{Diag}_i(\phi)) + \E_{\theta^{(w)}} \big[(Y_i - M_{ik}^{(w+1)})^\top \text{Diag}_i(\phi)^{-1}(Y_i - M_{ik}^{(w+1)}) \, | \, \mathcal{D}_n\big] \\
    & = \frac{1}{2} \sum_{\ell=1}^{L} n_i^{\ell} \log \phi_\ell + \sum_{\ell=1}^{L} \frac{1}{\phi_\ell}\E_{\theta^{(w)}} \big[ (Y_i^\ell - M_{ik}^{\ell (w+1)})^\top (Y_i^\ell - M_{ik}^{\ell (w+1)}) \, | \, \mathcal{D}_n\big].
\end{align*}
Then, the gradient of $A^{3}_i$ along $\phi_\ell$ is simply
\begin{equation*}
  \frac{\partial  A^{3}_i(\beta_k^{(w+1)}, \text{Diag}_i(\phi))}{\partial \phi_\ell} = \frac{n_i^\ell}{\phi_\ell}  - \frac{1}{\phi_\ell^2} \E_{\theta^{(w)}} \big[ \big(Y_i^\ell -  M_{ik}^{\ell (w+1)}\big)^\top \big(Y_i^\ell - M_{ik}^{\ell(w+1)}\big) \, | \, \mathcal{D}_n \big],
\end{equation*}
The closed-form update of $\phi_\ell^{(w)}$ is then obtained by setting
\begin{align*}
    & -n^{-1} \sum_{i=1}^n \sum_{k=1}^{K} \tilde\pi_{ik}^{\theta^{(w)}}  \frac{\partial A^{3}_i(\beta_k^{(w+1)}, \text{Diag}_i(\phi))}{\partial \phi_\ell} = 0\\
     \Leftrightarrow \quad &-n^{-1}\sum_{i=1}^n \sum_{k=1}^{K} \tilde{\pi}_{ik}^{\theta^{(w)}} \Big(\frac{n_i^\ell}{\phi_\ell}  - \frac{1}{\phi_\ell^2} \E_{\theta^{(w)}} \big[ \big(Y_i^\ell -  M_{ik}^{\ell (w+1)}\big)^\top \big(Y_i^\ell - M_{ik}^{\ell(w+1)}\big) \, | \, \mathcal{D}_n \big]\Big) = 0 \\
    \Leftrightarrow \quad & \sum_{i=1}^n \sum_{k=1}^{K} \tilde\pi_{ik}^{\theta^{(w)}} \Big(n_i^\ell \phi_\ell^{(w+1)} - \E_{\theta^{(w)}} \big[ \big(Y_i^\ell -  M_{ik}^{\ell (w+1)}\big)^\top \big(Y_i^\ell - M_{ik}^{\ell(w+1)}\big) \, | \, \mathcal{D}_n \big] \Big)= 0.
\end{align*}
The result follows from the fact that
\begin{align*}
 &\E_{\theta^{(w)}} \big[ \big(Y_i^\ell -  M_{ik}^{\ell (w+1)}\big)^\top \big(Y_i^\ell - M_{ik}^{\ell(w+1)}\big) \, | \, \mathcal{D}_n \big] \\
  & \quad = \E_{\theta^{(w)}} \big[\big(Y_i^\ell - U_i^\ell \beta_k^{\ell (w+1)}\big)^\top \big(Y_i^\ell - U_i^\ell \beta_k^{\ell(w+1)} -2V_i^\ell b_i^\ell\big) + b_i^\top V_i^\top V_i b_i \, | \, \mathcal{D}_n \big] \\
  & \quad =\big(Y_i^\ell - U_i^\ell {\beta_k^\ell}^{(w+1)}\big)^\top \big(Y_i^\ell - U_i^\ell {\beta_k^\ell}^{(w+1)}-2V_i^\ell \E_{\theta^{(w)}}[ b_i^\ell \, | \, \mathcal{D}_n]\big)+ \Tr\big({V_i^\ell}^\top V_i^\ell \E_{\theta^{(w)}}[ b_i^\ell b_i^{\ell\top} \, | \, \mathcal{D}_n]\big)\Big).
\end{align*}
\end{proof}

\subsection{M-step: Update $\xi$}
\label{sec:step-update-xi}
In $\cQ_n(\theta, \theta^{(w)})$, the parameter $\xi$ appears only in the term 
\begin{align*}
     &-n^{-1} \sum_{i=1}^{n} \sum_{k=1}^{K} \tilde\pi_{ik}^{\theta^{(w)}} A^{2}_{i, k}(\xi) \\
     & \quad = -n^{-1} \sum_{i=1}^{n} \sum_{k=1}^{K} \tilde\pi_{ik}^{\theta^{(w)}} \log \Big(\dfrac{e^{X_i^\top\xi_k}}{\sum_{j=1}^{K} e^{X_i^\top\xi_j}} \Big) \\
    & \quad = -n^{-1} \sum_{i=1}^{n} \bigg( \tilde\pi_{ik}^{\theta^{(w)}} \log \Big(1 + \sum_{\substack{ j \neq k \\ j=1 }}^{K}e^{X_i^\top(\xi_j - \xi_k)}\Big) + \sum_{\substack{m \neq k \\ m=1}}^{K} \tilde\pi_{im}^{\theta^{(w)}} \log \Big(1 + e^{X_i^\top(\xi_k - \xi_m)} + \sum_{\substack{j \neq k, j \neq m\\ j=1}}^{K}e^{X_i^\top(\xi_j - \xi_m)}\Big)\bigg).
\end{align*}
For $k \in \{1, \dots, K\}$, the update for $\xi_k^{(w)}$ requires to solve the minimization problem
\begin{equation}
  \label{eq:minimization-xi_k_ext}
  \xi_k^{(w+1)} \in \underset{\xi \in \R^p}{\argmin} \, \mathcal{F}_{1,k}(\xi) + \zeta_{1,k} \Omega_1(\xi),
\end{equation}
where $\mathcal{F}_{1,k}$ is defined by
\begin{align*}
\mathcal{F}_{1,k}(\xi) = n^{-1} \sum_{i=1}^n  \bigg( & \tilde\pi_{ik}^{\theta^{(w)}} \log \Big(1 + \sum_{\substack{ j \neq k \\j=1 }}^{K}e^{X_i^\top(\xi_j - \xi)}\Big) \\
& + \sum_{\substack{m \neq k \\ m=1}}^{K} \tilde\pi_{im}^{\theta^{(w)}} \log \Big(1 + e^{X_i^\top(\xi - \xi_m)} + \sum_{\substack{j \neq k, j \neq m\\ j=1}}^{K}e^{X_i^\top(\xi_j - \xi_m)}\Big)\bigg)
\end{align*}
and $\Omega_1$ is the elastic net regularization. We choose to solve~\eqref{eq:minimization-xi_k_ext} using the L-BFGS-B algorithm~\citep{zhu1997algorithm} which belongs to the class of quasi-Newton optimization routines and solves the given minimization problem by computing approximations of the inverse Hessian matrix of the objective function. It can deal with differentiable convex objectives with box constraints.

In order to use it with $l_1$ part of the elastic net regularization, which is not differentiable, we use the trick borrowed from \citet{andrew2007scalable}: for $a \in \R$, write $|a| = a^+ + a^-$, where $a^+$ and $a^-$ are respectively the positive and negative part of $a$, and add the constraints $a^+ \geq 0$ and $a^- \geq 0$.
Namely, we rewrite the minimization problem~\eqref{eq:minimization-xi_k_ext} as the following differentiable problem with box constraints
\begin{equation}
  \label{eq:sub-problem-xi_k}
  \begin{split}
    \text{minimize}& \quad \quad \mathcal{F}_{1,k}(\xi^+ - \xi^-) + \zeta_{1,k} \big((1 - \eta) \sum_{j=1}^p (\xi_j^+ + \xi_j^-) + \dfrac \eta 2 \norm{{\xi}^+ - {\xi}^-}_2^2 \big) \\
    \text{subject to}& \quad \quad \xi_j^+ \geq 0 \text{ and } \xi_j^- \geq 0 \text{ for } j \in \{1, \dots, p\}
  \end{split} 
\end{equation}
where $\xi^\pm = (\xi_1^\pm, \ldots, \xi_p^\pm)^\top$. The L-BFGS-B solver requires the exact value of the gradient, which is easily given by
\begin{align}
  \label{eq:grad-xi_k}
  \frac{\partial \mathcal{F}_{1,k}(\xi)}{\partial \xi} = -n^{-1} \sum_{i=1}^n  \big(\tilde\pi_{ik}^{\theta^{(w)}} - \dfrac{e^{X_i^\top\xi}}{e^{X_i^\top\xi} + \sum\limits_{\substack{j \neq k\\ j=1}}^{K}e^{X_i^\top\xi_j}} \big) X_i^\top.
\end{align}
In practice, we use the \texttt{Python} solver \texttt{fmin\_l\_bfgs\_b} from \texttt{scipy.optimize}~\citep{virtanen2020scipy}.

\subsection{M-step: Update $\gamma$}
\label{sec:M-step-update-gamma}

In $\cQ_n(\theta, \theta^{(w)})$, for $k \in \{1, \dots, K\}$, $\gamma_k$ appears only in the term 
\begin{align*}
     &-n^{-1} \sum_{i=1}^{n} \tilde\pi_{ik}^{\theta^{(w)}} A^{4}_i(\gamma_k, \lambda_0) \\
     & \quad = -n^{-1} \sum_{i=1}^{n} \tilde\pi_{ik}^{\theta^{(w)}} \Big( \Delta_i \big(\log \lambda_0(T_i) + \psi_i(T_i)^\top \gamma_k \big) - \sum_{j=1}^{J} \lambda_0(\tau_j) \exp \big(  \psi_i(\tau_j)^\top \gamma_k \big) \ind{\{\tau_j \leq T_i\}} \Big) \\
    & \quad = -n^{-1} \sum_{i=1}^n \tilde\pi_{ik}^{\theta^{(w)}} \Big(\Delta_i \psi_i(T_i)^\top \gamma_k - \sum_{j=1}^J \lambda_0(\tau_j) \exp \big(\psi_i(\tau_j)^\top \gamma_k\big) \ind{\{\tau_j \leq T_i\}} \Big) + \text{constants}.
\end{align*}
Then the update for $\gamma_k^{(w)}$ requires to solve the following minimization problem
\begin{equation}
  \label{eq:minimization-gamma_k_ext}
\gamma_k^{(w+1)} \in \underset{\gamma \in \R^{LM}}{\argmin} \, \mathcal{F}_{2,k}(\gamma) + \zeta_{2,k} \Omega_2(\gamma), 
\end{equation}
where $\mathcal{F}_{2,k}$ is defined by
\[\mathcal{F}_{2,k}(\gamma) = -n^{-1} \sum_{i=1}^n \tilde\pi_{ik}^{\theta^{(w)}} \Big(\Delta_i \psi_i(T_i)^\top \gamma - \sum_{j=1}^J \lambda_0^{(w)}(\tau_j) \exp \big(\psi_i(\tau_j)^\top \gamma\big) \ind{\{\tau_j \leq T_i\}} \Big) \]
and $\Omega_2$ is the sparse group lasso regularization. We choose to solve problem~\eqref{eq:minimization-gamma_k_ext} using the iterative soft-thresholding algorithm (ISTA), which is a proximal gradient descent algorithm~\citep{beck2009fast}. In our context, this method requires the gradient of $\mathcal{F}_{2,k}$ as well as the proximal operator~\citep{moreau1962fonctions} of the sparse group lasso. We refer to the proof of~\citet[Theorem 1]{yuan2011efficient} to show that the proximal operator of the sparse group lasso can be expressed as the composition of the group lasso and the lasso proximal operators, which are both well known analytically~\citep{bach2012optimization} and tractable. The gradient of $\mathcal{F}_{2,k}$ is here given by
\begin{equation}
    \label{eq:grad-gamma-k}
  \frac{\partial \mathcal{F}_{2,k}(\gamma)}{\partial \gamma} = -n^{-1} \sum_{i=1}^n \tilde\pi_{ik}^{\theta^{(w)}} \Big( \Delta_i - \sum_{j=1}^J \lambda_0^{(w)}(\tau_j) \exp \big(\psi_i(\tau_j)^\top \gamma \big) \ind{\{T_i \geq \tau_j\}} \Big) \psi_i(T_i)^\top.
\end{equation}
We use the \texttt{Python} library \texttt{copt} for the implementation of proximal gradient descent~\citep{copt} and we propose in our \texttt{FLASH} package a first \texttt{Python} implementation for the proximal operator of the sparse group lasso.

\subsection{Convex optimization problems with respect to $\xi$ and $\gamma$}
\label{sec:convex-prb}
\begin{lemma}
\label{lemma:strictly-convex} The optimization problems defined in \eqref{eq:minimization-xi_k_ext} and \eqref{eq:minimization-gamma_k_ext} are convex.
\end{lemma}

\begin{proof}
\label{sec:proof-lemma-convex}
Given first order derivative of  $\mathcal{F}_{1,k}(\xi)$ in $\eqref{eq:grad-xi_k}$, we show that the second order derivative of $\mathcal{F}_{1,k}(\xi)$ is positive definite
\begin{equation*}
  \frac{\partial^2 \mathcal{F}_{1,k}(\xi)}{\partial \xi \partial \xi^\top} = n^{-1} \sum_{i=1}^n \Big(\dfrac{e^{X_i^\top\xi}}{e^{X_i^\top\xi} + \sum\limits_{\substack{j \neq k\\ j=1}}^{K}e^{X_i^\top\xi_j}}\Big)\Big(1 - \dfrac{e^{X_i^\top\xi}}{e^{X_i^\top\xi} + \sum\limits_{\substack{j \neq k\\ j=1}}^{K}e^{X_i^\top\xi_j}}\Big) X_iX_i^\top \in S_{++}^p.
\end{equation*}

Given first order derivative of  $\mathcal{F}_{2,k}(\gamma)$ in $\eqref{eq:grad-gamma-k}$, we show that the second order derivative of $\mathcal{F}_{2,k}(\gamma)$ is positive definite
\begin{equation*}
  \frac{\partial^2 \mathcal{F}_{2,k}(\gamma)}{\partial \gamma \partial \gamma^\top} = n^{-1} \sum_{i=1}^n \Big( \tilde\pi_{ik}^{\theta^{(w)}} \sum_{j=1}^J \lambda_0^{(w)}(\tau_j) \exp \big(\psi_i(\tau_j)^\top \gamma \big) \ind{\{T_i \geq \tau_j\}}\Big) \psi_i(T_i) \psi_i(T_i)^\top \in S_{++}^{LM}.
\end{equation*}
As we already defined the elastic net and sparse group lasso regularization
\[\Omega_1(\xi) =  (1-\eta)\norm{\xi}_1 + \dfrac\eta2 \norm{\xi}_2^2 \quad \text{and} \quad \Omega_2(\gamma) = (1-\tilde{\eta})\norm{\gamma}_1 + \tilde{\eta} \sum_{\ell=1}^L \norm{\gamma^\ell}_2, \]
with $(\eta, \tilde{\eta}) \in [0, 1]^2$. Note that every norm is convex and a non-negative weighted sum of convex functions is convex \citep{boyd2004convex} then $\Omega_1(\xi)$ and  $\Omega_2(\gamma)$ are convex functions. Therefore, $\mathcal{F}_{1,k}(\xi) + \Omega_1(\xi)$ and $\mathcal{F}_{2,k}(\gamma) + \Omega_2(\gamma)$ are strictly convex.
\end{proof}

\subsection{The extended EM algorithm}
Algorithm~\ref{alg:ext-EM} below describes the main steps of our proposed EM algorithm.
\begin{algorithm}
\caption{The extended EM algorithm for FLASH inference} \label{alg:ext-EM}
\begin{algorithmic}[1]
    \Data Training data $\cD_n$; tuning hyper-parameters $(\zeta_{1,k}, \zeta_{2,k})_{k \in \{1,\ldots,K\}}$
    \Input maximum iteration $W$, tolerance $\varepsilon$
    \Output Last parameters $\hat \theta \in \R^P$
    \State Initialize parameters $\theta^{(0)} \in \R^P$\; 
    \For{$w = 1,\ldots, W$}
    \Statex \hspace{.48cm}{\textbf{E-step:}}
    \State Compute $(\E_{\theta^{(w)}}[ b_i | T_i, \Delta_i, Y_i])_{i \in \{1, \dots, n\}}$, $(\E_{\theta^{(w)}}[ b_i b_i^\top | T_i, \Delta_i, Y_i])_{i \in \{1, \dots, n\}}$, $(\widetilde{\pi}_{ik}^{\theta^{(w)}})_{\substack{i \in \{1, \dots, n\} \\ k \in \{1, \dots, K\}}}$
    \State  Compute $(\E_{\theta^{(w)}}[ b_i b_i^\top | T_i, \Delta_i, Y_i])_{i \in \{1, \dots, n\}}$
    \State{Compute $(\widetilde{\pi}_{ik}^{\theta^{(w)}})_{\substack{i \in \{1, \dots, n\} \\ k \in \{1, \dots, K\}}}$}
    \Statex \hspace{.48cm}{\textbf{M-step:}}
    \State Update $D^{(w+1)}$\;
    \State Update $(\xi_k^{(w+1)})_{k \in \{1, \dots, K\}}$ with L-BFGS-B\;
    \State Update $(\beta_k^{(w+1)})_{k \in \{1, \dots, K\}}$ \;
    \State Update $(\gamma_k^{(w+1)})_{k \in \{1, \dots, K\}}$ with proximal gradient descent \;
    \State Update $\lambda_0^{(w+1)}$ and $\phi^{(w+1)}$ \;
    \If{$\big(\cL^\text{pen}_n(\theta^{(w+1)}) - \cL^{\text{pen}}_n(\theta^{(w)}) \big)/\cL^\text{pen}_n(\theta^{(w)}) < \varepsilon$}
        \Statex \qquad \quad \textbf{break}
    \EndIf
    \EndFor
    \State \textbf{Return} {$\hat \theta = \theta^{(w+1)}$}
\end{algorithmic}
\end{algorithm}
\subsection{Monotone convergence}
\label{sec:monotone_converge}

By denoting $\zeta_p^{(w)} = (\theta_1^{(w+1)}, \dots, \theta_p^{(w+1)}, \theta_{p+1}^{(w)}, \dots, \theta_P^{(w)})$ and from the definition of our proposed EM algorithm, at step $w+1$, we have
    \begin{equation}
        \cQ^\text{pen}_n(\theta^{(w)}, \theta^{(w)}) \geq  \cQ^\text{pen}_n(\zeta_1^{(w)}, \theta^{(w)}) \geq \dots \geq \cQ^\text{pen}_n(\zeta_{P-1}^{(w)}, \theta^{(w)})
        \geq \cQ^\text{pen}_n(\theta^{(w+1)}, \theta^{(w)}).
    \end{equation}
Then, we are in a generalized EM (GEM) setting~\citep{dempster1977maximum}, where
\begin{equation}
    \cQ^\text{pen}_n(\theta^{(w + 1)}, \theta^{(w)}) \leq 
    \cQ^\text{pen}_n(\theta^{(w)}, \theta^{(w)}).
\end{equation}
For such algorithms, we refer to monotonicity of the likelihood property from the book of \citet[Section 3.3]{mclachlan2007algorithm} to show that the objective function~\eqref{eq:pen-log-lik_ext} decreases at each iteration, namely 
\[\cL^{\text{pen}}_n(\theta^{(w+1)}) \leq \cL^{\text{pen}}_n(\theta^{(w)}).\]

\section{Mathematical details of JLCMs and SREMs}
\label{sec:SREM_JLCM}
Given our notation used in the main body of the paper, we define here the sub-models of the two main approaches of joint models: JLCMs and SREMs.

\paragraph{JLCMs}
It assumes that the population is heterogeneous and that there are homogeneous latent classes that share the same marker trajectories and the same prognosis. The latent class membership probability is assumed to take the form of multinomial logistic regression
\begin{equation*}
  \P[g_i=k \, | \, X_i] = \dfrac{e^{X_i^\top\xi_k}}{\sum_{j=1}^{K}e^{X_i^\top\xi_j}}.
\end{equation*}
The dependence between the time-to-event and the longitudinal
marker is fully captured by a latent class structure. There are no shared associations between the longitudinal and survival models. Given the latent class membership, each submodel is assumed to be independent. If we choose Gaussian linear model for longitudinal markers and Cox relative risk model for the time-to-event, we have
\[y_i^{\ell}(t^{\ell}_{ij}) \, | \,  b_i^\ell, \, g_i=k \sim \mathcal{N}(m_{ik}^{\ell}(t^{\ell}_{ij}), \phi_\ell) \quad \text{and} \quad \lambda(t \, | \, g_i = k) = \lambda_0(t) \exp \Big(X_i^\top {\gamma_k}\Big),\]
where $\gamma_k$ is the $p$-vector of unknown parameters.
We consider the implementation of JLCMs in R package \texttt{LCMM} (function \texttt{mpjlcmm}). In this context, the predictive marker for subject $i$ at time $s_i$ is
\begin{equation*}
  \widehat \cR_{ik}(s_i) = \dfrac{\P_{\hat{\theta}}(g_i=k) f_{\hat{\theta}}(T_i = s_i, \Delta_i = 0, \cY_i(s_i^{-}) | g_i=k)}{\sum_{j=1}^{K} \P_{\hat{\theta}}(g_i=j) f_{\hat{\theta}}(T_i = s_i, \Delta_i = 0, \cY_i(s_i^{-}) | g_i=j)},
\end{equation*}
where the density $f_{\hat{\theta}}$ are the one corresponding to a JLCM model.

\paragraph{SREMs} It assumes a homogeneous population of subjects and the dependency between the time-to-event and the longitudinal
marker is influenced by some random effects learned in a linear mixed model. If we choose Gaussian linear model for longitudinal markers, we have
\begin{equation*}
    y_i^{\ell}(t^{\ell}_{ij}) \, | \,  b_i^\ell \sim \mathcal{N}(m_i^{\ell}(t^{\ell}_{ij}), \phi_\ell),
\end{equation*}
where $m_i^\ell(t^{\ell}_{ij}) = u^\ell(t^{\ell}_{ij})^\top\beta^\ell + v^\ell(t^{\ell}_{ij})^\top b_i^\ell$ and $\beta^\ell$ is a $r_\ell$-vector of unknown fixed effect parameters. The random effects are included as covariates in the survival model through the shared association functions $\phi$. If we choose Cox relative risk model for the time-to-event, we have
\begin{equation*}
  \lambda(t \, | \, g_i = k) = \lambda_0(t) \exp \Big(X_i^\top \gamma_0 + \sum_{\ell=1}^L \phi(b_i^\ell, t_i)^\top {{\gamma^\ell}}\Big),
\end{equation*}
where $\gamma_0$ and $\gamma^\ell$ are unknown parameters.

We consider the implementation of SREMs in R package \texttt{JMbayes}. In this context, the predictive marker for subject $i$ at time $s_i$ is

\[\widehat \cR_i(s_i) = \exp \Big( X_i^\top \gamma_0 + \sum_{\ell=1}^L \phi(b_i^\ell, s_i)^\top {{\gamma^\ell}} \Big),\] 
where the shared associations takes the form $\phi(b_i^\ell, s_i) = u^\ell(s_i)^\top\beta^\ell + v^\ell(s_i)^\top b_i^\ell$.

\section{Experimental details and additional experiments}
\label{sec:experiment_detail}
\subsection{Initialization}
\label{sec:initialization}

\subsubsection{Initialization}
In order to help convergence,  $\theta^{(0)}$ should be well chosen. We then give some details about the starting point $\theta^{(0)}$ of this algorithm. For all $k = 1, \ldots, K$, we first choose $\xi_k^{(0)} = \mathbf{0}_d$ and ${\gamma_k}^{(0)} = 0.01 * \mathbf{1}_{LA}$. Then, we initialize $\lambda_0^{(0)}$ like if there is no latent classes ($\gamma_{1}^{(0)} = \cdots = \gamma_{K}^{(0)}$) with a standard Cox proportional hazards regression with time-independent features. Finally, the longitudinal submodel parameters $\beta_k^{(0)}$, $D^{(0)}$ and $\phi^{(0)}$ are initialized -- again like if there is no latent classes ($\beta_1^{(0)} = \cdots = \beta_{K}^{(0)}$) -- using a multivariate linear mixed model with an explicit EM algorithm, being itself initialized with univariate fits.

\subsubsection{Multivariate linear mixed model}
\label{sec:MLMM}

Let us derive here the explicit EM algorithm for the multivariate Gaussian linear mixed model used to initialized the longitudinal parameters $\beta_k^{(0)}$, $D^{(0)}$ and $\phi^{(0)}$ in the proposed EM algorithm in Section~\ref{sec:initialization}, acting as if there is no latent classes ($\beta_1^{(0)} = \cdots = \beta_{K}^{(0)}$).
For the sake of simplicity, let us denote here 
\[\theta = (\beta^\top, D, \phi^\top)^\top \in \R^P\]
the parameter vector to infer.
The conditional distribution of $Y_i|b_i$ then writes
\[f(Y_i|b_i ; \theta) = -(2\pi)^{-\frac{n_i}{2}} \det(\Sigma_i)^{-\frac{1}{2}} \exp^{-\frac{1}{2}(Y_i - M_i)^\top \Sigma_i^{-1}(Y_i - M_i)},\]
where ${M_i = U_i\beta + V_ib_i}$ in this context. The negative complete log-likelihood then writes
\begin{align*}
\cL_n^\text{comp}(\theta) &= \cL_n^\text{comp}(\theta ; \cD_n, \textbf{\textit{b}}) \\
&= \sum_{i=1}^n 
-\dfrac12 \big(n_i \log 2\pi + \log \det(\Sigma_i) + (Y_i - M_i)^\top \Sigma_i^{-1}(Y_i - M_i)\big)\\
& \quad \quad \quad -\dfrac12 \big(r \log 2\pi + \log \det(D) + b_i^\top D^{-1}b_i\big).
\end{align*}
  
\textit{E-step.}
Supposing that we are at step $w + 1$ of the algorithm, with current iterate denoted $\theta^{(w)}$, we need to compute the expectation of the negative complete log-likelihood conditional on the observed data and the current estimate of the parameters, which is given by 
\[\cQ_n(\theta, \theta^{(w)}) = \E_{\theta^{(w)}}[\cL_n^\text{comp}(\theta) | \cD_n].\]
Here, the calculation of this quantity is reduced to the calculation of $\E_{\theta^{(w)}}[b_i | Y_i]$ and $\E_{\theta^{(w)}}[b_ib_i^T | Y_i]$ for $i = 1,\ldots,n$.
The marginal distributions of $Y_i$ and $b_i$ being both Gaussian, one has from Bayes Theorem
\[ f(b_i|Y_i; \theta^{(w)}) \propto \exp \big( -\dfrac12 (b_i - \mu_i^{(w)})^\top {\Omega_i^{(w)}}^{-1} (b_i - \mu_i^{(w)}) \big)\]
where
\[\Omega_i^{(w)} = ({V_i}^\top {\Sigma_i^{(w)}}^{-1}V_i + {D^{(w)}}^{-1})^{-1} \quad \text{ and } \quad \mu_i^{(w)} = \Omega_i^{(w)}{V_i}^\top {\Sigma_i^{(w)}}^{-1}(Y_i - U_i\beta^{(w)}).\]
Then, one has
\begin{center}
  $\left\{
    \begin{aligned}
      &\E_{\theta^{(w)}}[b_i | Y_i] = \mu_i^{(w)}, \\
      &\E_{\theta^{(w)}}[b_ib_i^T | Y_i] = \Omega_i^{(w)} + \mu_i^{(w)}{\mu_i^{(w)}}^\top.
    \end{aligned}
    \right.$
\end{center}

\textit{M-step.}
Here, we need to compute \[\theta^{(w+1)} \in \argmin_{\theta \in \R^P} \cQ_n(\theta, \theta^{(w)}) .\]
The parameters updates are then naturally given in closed form by zeroing the gradient. One obtains
\[\beta^{(w+1)} = \Big(\sum_{i=1}^n U_{i}^\top U_{i}\Big)^{-1}  \Big(\sum_{i=1}^n U_{i}^\top  Y_i - U_{i} V_{i} \E_{\theta^{(w)}}[ b_i | Y_i] \Big),\]
\begin{align*}
\phi_\ell^{(w+1)} = \Big(\sum_{i=1}^n n_i^\ell \Big)^{-1}  \Big(\sum_{i=1}^n \big(&Y_i^\ell - U_i^\ell\beta_\ell^{(w+1)}\big)^\top \big(Y_i^\ell - U_i^\ell\beta_\ell^{(w+1)} - 2V_i^\ell\E_{\theta^{(w)}}[b_i^\ell | Y_i^\ell]\big) \\
&+\Tr\big({V_i^\ell}^\top V_i^\ell \E_{\theta^{(w)}}[ b_i^\ell b_i^{\ell\top} | Y_i^\ell]\big) \Big)
\end{align*}
and \[D^{(w+1)} = n^{-1} \sum_{i=1}^n \E_{\theta^{(w)}}[ b_i {b_i}^\top | Y_i].\]

\subsubsection{Implementation of multivariate linear mixed model}
\label{sec:MLMM class}

We implement an EM algorithm for fitting a multivariate linear mixed model used to initialize parameters of longitudinal submodel in Algorithm~\ref{alg:ext-EM}. Let us introduce the list $\Omega^{(w)} = [\Omega_1^{(w)}, \ldots, \Omega_n^{(w)}]$, the matrices $\mu = [\mu_1, \ldots, \mu_n] \in \R^{r \times n}$, $U^\ell = {[{U_1^\ell}^\top \cdots {U_n^\ell}^\top]}^\top \in \R^{n_\ell \times q_\ell}$, $U = {[U_1^\top \cdots U_n^\top]}^\top \in \R^{\cN \times q}$,
\[ V^\ell = \begin{bmatrix}
    V_1^\ell & \cdots & 0\\
    \vdots &  \ddots & \vdots \\
    0 & \cdots & V_n^\ell
\end{bmatrix},
\quad
V = \begin{bmatrix}
  V_{1} & \cdots & 0\\
  \vdots &  \ddots & \vdots \\
  0 & \cdots & V_{n}
\end{bmatrix},
\quad
{\Omega^\ell}^{(w)} = \begin{bmatrix}
    {\Omega_{1}^\ell}^{(w)} & \cdots & 0\\
    \vdots &  \ddots & \vdots \\
    0 & \cdots & {\Omega_{n}^\ell}^{(w)}
\end{bmatrix}\]
that belong respectively in $\R^{n_\ell \times nr_\ell}$, $\R^{\cN \times nr}$ and $\R^{nr_\ell \times nr_\ell}$, and the vectors $\tilde \mu^{(w)} = ({\mu_1^{(w)}}^\top \cdots {{\mu_n^{(w)}}^\top)}^\top \in \R^{nr}$, $(\tilde \mu^\ell)^{(w)} = \big({(\mu_1^\ell)^{(w)}}^\top \cdots {(\mu_n^\ell)^{(w)}}^\top \big)^\top \in \R^{nr_\ell}$, $y^\ell = ({y_1^\ell}^\top  \cdots {y_n^\ell}^\top)^\top \in \R^{n_\ell}$
with $n_\ell = \sum_{i=1}^{n} n_i^\ell$ and $y = (y_1^\top \cdots y_n^\top)^\top \in \R^\cN$ with $\cN = \sum_{i=1}^{n} n_i$.
The $\beta$ update then rewrites
\[\beta^{(w+1)} = ({U}^\top U)^{-1}{U}^\top(y-V \tilde \mu^{(w)}).\]
For the $D$ update, one has
\[D^{(w+1)} = n^{-1} \big( \text{sum}(\Omega^{(w)}) + \mu^{(w)}{\mu^{(w)}}^\top\big).\]
And finally for the $\phi$ update, one has
\begin{align*}
\phi_\ell^{(w+1)} = n_\ell^{-1} \big[&(y^\ell - U^\ell{\beta_\ell^{(w+1)})}^\top \big(y^\ell - U^\ell\beta_\ell^{(w+1)} - 2 V^\ell (\tilde\mu^\ell)^{(w)}\big) \\
&+ \text{Tr}\big\{{V^\ell}^\top V^\ell \big({\Omega^\ell}^{(w)} + (\tilde\mu^\ell)^{(w)} {(\tilde\mu^\ell)^{(w)}}^\top\big)\big\} \big].
\end{align*}
In our implementation, these parameters are initialized with univariates fits using the function \texttt{mixedlm} (linear mixed effects model) in the \texttt{Python} package \texttt{statsmodels}.

\subsection{Details of the simulation setting}
\label{sec:data-simu}
We assume that each of the $n$ subjects belongs to one of two different profiles: high-risk and low-risk. 
Let us denote by $\mathcal{H} \subset \{1, \dots, n \}$ the set of high-risk subjects. For generating the time-independent features matrix of subject $i$, we take 
\[X_i \in \R^{p} \sim 
  \begin{cases}
    \cN \big(\mu, \bSigma_1(\rho_1)\big) &\text{if } i \notin \cH, \\ 
    \cN \big(-\mu, \bSigma_1(\rho_1)\big) &\text{if } i \in \cH,
  \end{cases}
 \]
where the mean $\mu$ corresponds to the gap between time-independent features of high-risk subjects and low-risk subjects, and $\bSigma_1(\rho_1)$ a $p \times p$ Toeplitz covariance matrix~\citep{mukherjee1988some} with correlation $\rho_1 \in (0, 1)$, that is, $\bSigma_1(\rho_1)_{jj'} = \rho_1^{|j - j'|}$. In order to simulate the class $g_i$ for each subject $i$, we choose a sparse coefficient vector where we decide to keep only $\bar{p}$ active features, that is
\begin{equation}
  \label{eq:xi-sparse}
  \xi = (\nu,\ldots,\nu ,0,\ldots,0) \in \R^p,
\end{equation} 
with $\nu\in\R$ being the value of the active coefficients. Then, we generate $g_i \sim \cB\big(\pi_\xi(X_i)\big)$, where $\cB(\alpha)$ denotes the Bernoulli distribution with parameter $\alpha \in [0,1]$ and 
\begin{equation*}
    \pi_\xi(X_i) = \dfrac{e^{X_i^\top\xi}}{1 + e^{X_i^\top\xi}}
\end{equation*}

Now, concerning the simulation of longitudinal markers, the idea is to sample from multivariate normal distributions. Moreover, we want to induce sparsity in the longitudinal data to reduce correlation between longitudinal features in each class $k$. We denote by $\cS_k$ the set of active longitudinal features in class $k$, which is randomly selected from the set $\{1, \dots, L \}$. Then, we simulate longitudinal features of the form 
\[ Y_i^\ell(t) = \sum_{k=1}^{K} \ind{\{g_i=k\}} \Big( \big((1, t)^\top \beta_k^\ell + (1, t)^\top b_i^\ell\big) \ind{\{\ell \in \cS_k\}} + \varepsilon_i^\ell(t) \Big) \]
where $t \geq 0$, the error term $\varepsilon_i^\ell(t) \sim \cN(0, \sigma_\ell^2)$, the global variance-covariance matrix for the random effects components is such that $D = \bSigma_2(\rho_2)$, a $r \times r$ Toeplitz covariance matrix with correlation $\rho_2 \in (0, 1)$, and the fixed effect parameters are generated according to 
\[ \beta_k^\ell \sim \cN\Big( 
\mu_k, 
\begin{bmatrix}
  \rho_3 & 0\\
  0 & \rho_3
\end{bmatrix} 
\Big) \]
for $k \in \{1, 2\}$ and with correlation $\rho_3 \in (0, 1)$. The number of observations for each subject is randomly selected from 1 to 10, and the measurement times are simulated from a uniform distribution with mininmum zero and maximum its survival time.

Now to generate survival times,  we choose a risk model with
\begin{equation}
  \lambda_i(t \, | \, g_i = k) = \lambda_0(t) \exp \Big(\sum_{\ell=1}^L \Psi_{i, k}^\ell(t){\gamma_k^\ell} \Big),
\end{equation}
We choose a Gompertz distribution~\citep{gompertz1825xxiv} for the baseline, that is
\begin{equation}
  \label{eq:baseline}
  \lambda_0(t) = \kappa_1 \kappa_2 \exp(\kappa_2t)
\end{equation}
with $\kappa_1 > 0$ and $\kappa_2 \in \R$ the scale and shape parameters respectively, which is a common distribution choice in survival analysis~\citep{klein2005survival} with a rich history in describing mortality curves. For the choice of the association features, we consider the two functionals in form of random effects linear predictor $m_k^\ell(t)$~\citep{chi2006joint} and random effects $b^\ell$~\citep{hatfield2011joint}, that is $\Psi_{i, k}^\ell(t) = \big(\beta_{k,1}^\ell + \beta_{k,2}^\ell t + b_{i,1}^\ell + b_{i,2}^\ell t\ , b_i^{\ell\top} \big)^\top$ and $\gamma_k^\ell = \nu_k \ind{\{\ell \in \cS_k \}}$
 where $\nu_k \in \R$ is the coefficients of active association features for each group $k$. Then one can write
\[ \lambda_i(t|g_i = k) = \lambda_0(t) \exp ( \iota_{i,k,1} + \iota_{i,k,2} t ),\]
being a Cox model with a linear relationship between time-varying feature and log hazard that allows the following explicit survival times generation process. One can now generate survival times explicitly, via the inversion method, as
\begin{equation}
  \label{eq:time-generation}
  T_i^\star | g_i=k \sim \dfrac{1}{\iota_{i,k,2} + \kappa_2} \log \Big(1 - \dfrac{(\iota_{i,k,2} + \kappa_2) \log U_i}{\kappa_1 \kappa_2 \exp\iota_{i,k,1}} \Big)
\end{equation}
where $U_i \sim \cU\big([0,1]\big)$, see~\cite{austin2013correction}.
The distribution of the censoring variable $C_i$ is the geometric distribution $\cG(\alpha_c)$, where $\alpha_c \in (0, 1)$ is empirically tuned to maintain a desired censoring rate $r_c \in [0,1]$. 
The choice of all hyper-parameters is driven by the applications on simulated data presented in Section~\ref{sec:simu-1} in the main paper, and summarized in Table~\ref{table:parameters choice}. 

\begin{table}[htb]
\caption{Hyper-parameter choices for simulation with $n = 500$, $L = 5$ and $p=10$.}
\label{table:parameters choice}
\centering
\begin{tabular}{cccccccccccc}
\toprule
$|\cH|$ & $|\cS_k|$ & $(\rho_1, \rho_2, \rho_3)$ & $\mu$ & $\mu_1$ & $\mu_2$ & $\sigma_\ell^2$ & $(\kappa_1, \kappa_2)$ & $(\nu, \nu_1, \nu_2)$ & $r_c$ & $\bar{p}$\\
\midrule
200 & 2 & $(0.5, 0.01, 0.01)$ & 1 & 
$\begin{pmatrix}
  -0.6 \\ 0.1
\end{pmatrix}$ & 
$\begin{pmatrix}
  0.05 \\ 0.2
\end{pmatrix}$ 
& 0.25 & $(0.05, 0.1)$ & (0.2, 0.1, 0.4) & 0.3 & 5\\
\bottomrule
\end{tabular}
\end{table}

\subsection{Description of the datasets used in comparison study}
\label{sec:compared_data}
\paragraph{JoineRML simulation}
We use the classical R package \texttt{joineRML} \citep{hickey2018joinerml} to simulate multivariate longitudinal and time-to-event data from a joint model. The multivariate longitudinal features are generated for all possible measurement times using multivariate Gaussian linear mixed model. Failure times are simulated from proportional hazards time-to-event models. We sample two time-independent features and two longitudinal features for 250 individuals. 

\paragraph{PBCseq dataset} This dataset which is available in the R package \texttt{JMbayes} \citep{2017_JMBayes}, contains the follow-up of 312 patients with primary biliary cirrhosis, a rare autoimmune liver disease. Several longitudinal features are measured over time (for example serum bilirubin, serum cholesterol, albumin), along with information on gender, age, and drug used recorded once at the beginning of the study. Time-to-event is also recorded with a censoring rate of 55\%.

\paragraph{Aids dataset} This dataset which is available in the R package \texttt{JMbayes} \citep{2017_JMBayes}, compares the efficacy and safety of two drugs for 467 patients diagnosed with HIV who were either intolerant or resistant to zidovudine therapy. Information on gender, age, drug used, AIDS infection status, and level of intolerance to zidovudine is collected at the start of the study. The longitudinal feature of interest here is the measurement of the number of CD4 cell (a type of white blood cell), a laboratory test used to understand the progression of HIV disease. Time-to-event is also recorded with a censoring rate of 40\%.
\subsection{Procedure to evaluate model performance}
\label{sec:evaluation-procedure}

Let us describe now in Algorithm~\ref{alg:evaluation} the procedure we follow to evaluate model performance on simulated data and real data described in Section~\ref{sec:experiment} in the main paper.

\begin{algorithm}
\caption{Procedure followed to assess performances of a given model in our real-time prediction paradigm.} \label{alg:evaluation}
\begin{algorithmic}[1]
    \Input Dataset $\cD_n$; a \texttt{model} under study
    \Output Confidence intervals on C-index metric as well as on running time.
    \State We run $K^{\text{iter}}=50$ independent experiments
    \For{$k = 1,\ldots,K^{\textnormal{iter}}$}
        \State $\texttt{start\_time} = \texttt{time}()$
        \State $(\cD_{n_{\text{train}}}, \cD_{n_{\text{test}}}) \leftarrow \texttt{split\_train\_test}(\cD_n)$
        \State \texttt{model.fit}$(\cD_{n_{\text{train}}})$ 
        \For{$i = 1,\ldots,n_{\text{test}}$}
            \State $s_i \sim \underset{\ell \in \{1, \dots, L \} }{\max}(t_{in_i}^\ell) \times \big(1 - \text{Beta}(2, 5)\big)$ \State $Y_i \leftarrow \big(Y_i^\ell(t_{ij}^\ell) \big)_{\substack{j\in \{1, \ldots, n_i^\ell-1\} \\ \ell \in \{1, \ldots, L\} \ \ \ }} \text{ with } t_{ij}^\ell \leq s_i$ 
            \State $\mathcal{X}_i = (X_i, Y_i)$
            \State $\widehat \cR_i^{k}(s_i) \leftarrow \texttt{model.predict}(\mathcal{X}_i)$
        \EndFor
        \State $\texttt{score}_{k} \leftarrow \texttt{c\_index}\big((\widehat \cR_i^{k}(t_i), T_i, \Delta_i)_{i=1,\ldots,n_{test}} \big)$ 
        \State $\texttt{end\_time} = \texttt{time}()$
        \State $\texttt{running\_time}_{k} = \texttt{end\_time} - \texttt{start\_time}$
    \EndFor
    \State \textbf{Return} {$\hat \theta = \theta^{(w+1)}$}
\end{algorithmic}
\end{algorithm}

\textit{Screening phase.} We use the multivariate Cox model and C-index metric for selecting the $M$ most important features from a specific set of feature extraction functions $\cF$. For each feature in $\cF$, we extract the $n \times L$ matrix from all $L$ longitudinal markers of all $n$ subjects. Then we fit this extracted matrix with all the survival times $T$ and censoring indicators $\Delta$ in the Cox model and use the C-index metric to evaluate the performance. Finally, we select $M$ features have the highest C-index values.

\subsection{Interpretation of the model on medical datasets}
\label{sec:medical}
Tables \ref{tab:coef_PBC} and  \ref{tab:coef_Sepsis} provide the estimated coefficients of FLASH on the PBCseq and Sepsis datasets. Coefficients are organized as follows: first, the time-independent parameters $\xi$, then the coefficients corresponding to the association features of the low-risk group $\gamma_{1,m}^\ell$ and finally the ones of the high-risk group $\gamma_{2,m}^\ell$. The initial values of the longitudinal markers are also considered as time-independent features in these experiments. Since each longitudinal marker is associated to a vector of association features, what we call ``coefficient'' is actually the Euclidean norm of the coefficients associated to these association features, that is, the norm of $(\gamma_{k,1}^\ell, \ldots, \gamma_{k,M}^\ell)$. In addition, to obtain standard errors for the coefficient estimates, we use a Bootstrap approach adapted to the presence of a Lasso penalty following \citet{chzhen2019lasso}. We first run the model with the Lasso penalty to get the support of estimated coefficients. We then rerun the model 10 times without the Lasso penalty on bootstrap samples with only features whose coefficients are in the support of the first run.

\begin{table}[h]
\centering
\caption{Estimated coefficients with standard errors}
\subfloat[PBCseq dataset]{
\label{tab:coef_PBC}
\begin{tabular}[b]{cc} 
 \toprule
 \textbf{Features} & \textbf{Coefficient} \\ 
 \midrule
 Drug & 0.0 \\
 Age & 0.0 \\ 
 Sex & 0.3459 ± 0.0037 \\ 
 Initial value of Serbilir & 0.3476 ± 0.0091 \\ 
 Initial value of Albumin & 0.0\\ 
 Initial value of SGOT & 0.3664 ± 0.0069\\
 Initial value of Platelets & 0.0 \\
 Initial value of Prothrombin & 0.3587 ± 0.0088\\
 Initial value of Alkaline & 0.332 ± 0.0042\\
 Initial value of SerChol & 0.2978 ± 0.0093\\
 \midrule
 Serbilir & 0.0341 ± 0.0161 \\ 
 Albumin & 0.0833 ± 0.0088\\ 
 SGOT & 0.0 \\
 Platelets & 0.0325 ± 0.0159\\
 Prothrombin & 0.0819 ± 0.0129\\
 Alkaline & 0.0807 ± 0.0092\\
 SerChol & 0.0 \\
 \midrule
 Serbilir & 0.0 \\ 
 Albumin & 0.0\\ 
 SGOT & 0.0 \\
 Platelets & 0.0 \\
 Prothrombin & 0.0 \\
 Alkaline & 0.2788 ± 0.0355\\
 SerChol & 0.0\\
 \bottomrule
\end{tabular}
}
\qquad
\subfloat[Sepsis dataset]{
\label{tab:coef_Sepsis}
\begin{tabular}[b]{cc} 
 \toprule
 \textbf{Features} & \textbf{Coefficient} \\ 
 \midrule
 Age & 0.152 ± 0.0414\\
 Gender & 0.0 \\
 Temp & 0.0 \\
 DBP & 0.0669 ± 0.1234\\
 BaseExcess & 0.2779 ± 0.0404\\
 HCO3 & 0.0 \\
 pH & 0.0 \\
 PaCO2 & 0.119 ± 0.0987\\
 BUN & 0.0 \\
 Calcium & 0.0 \\
 Chloride & 0.0 \\
 Creatinine & 0.0 \\
 Glucose & 0.0 \\
 Magnesium & 0.0 \\
 Phosphate & 0.0 \\
 Potassium & 0.0 \\
 Hct & 0.0 \\
 Hgb & 0.0 \\
 PTT & 0.0 \\
 WBC & 0.0 \\
 Platelets & 0.0 \\
 Initial value of HR & 0.0 \\
 Initial value of O2Sat & 0.1325 ± 0.1049\\
 Initial value of SBP & 0.6485 ± 0.0788 \\ 
 Initial value of Resp & 0.7821 ± 0.1001\\ 
 \midrule
 HR & 0.0 \\ 
 02Sat & 0.0\\ 
 SBP & 0.1111 ± 0.0211\\
 Resp & 0.1016 ± 0.0176\\
 \midrule
 HR & 0.0198 ± 0.0051\\ 
 02Sat & 0.0783 ± 0.0066\\ 
 SBP & 0.0849 ± 0.0064\\
 Resp & 0.0809 ± 0.0056\\
 \bottomrule
\end{tabular}
}
\end{table}

Note that with this approach the errors of estimation of the coefficients that are not in the support of the first run are then not evaluated. In other words, we do not evaluate the stability of our variable selection approach. To do this, we could follow the procedure of \citet{bach2008bolasso} who suggest running several runs of Lasso on a bootstrap sample, and then looking at the different support of the coefficients.

\subsection{Experiments on a high-dimensional dataset}
\label{sec:NASA}
We evaluate the performance of FLASH model with a challenging high-dimensional dataset from NASA, which is available at \texttt{https://data.nasa.gov}. This dataset describes the degradation of 200 aircraft engines, where 17 multivariate longitudinal features are measured for each different aircraft engine until its failure. There are also three operational settings that significantly affect engine performance. Note that we only apply FLASH to this dataset because the other models did not converge after running for one day, highlighting the fact that they do not scale to high-dimensional settings.

\begin{figure}[!htb]
    \centering
    \includegraphics[scale=.4]{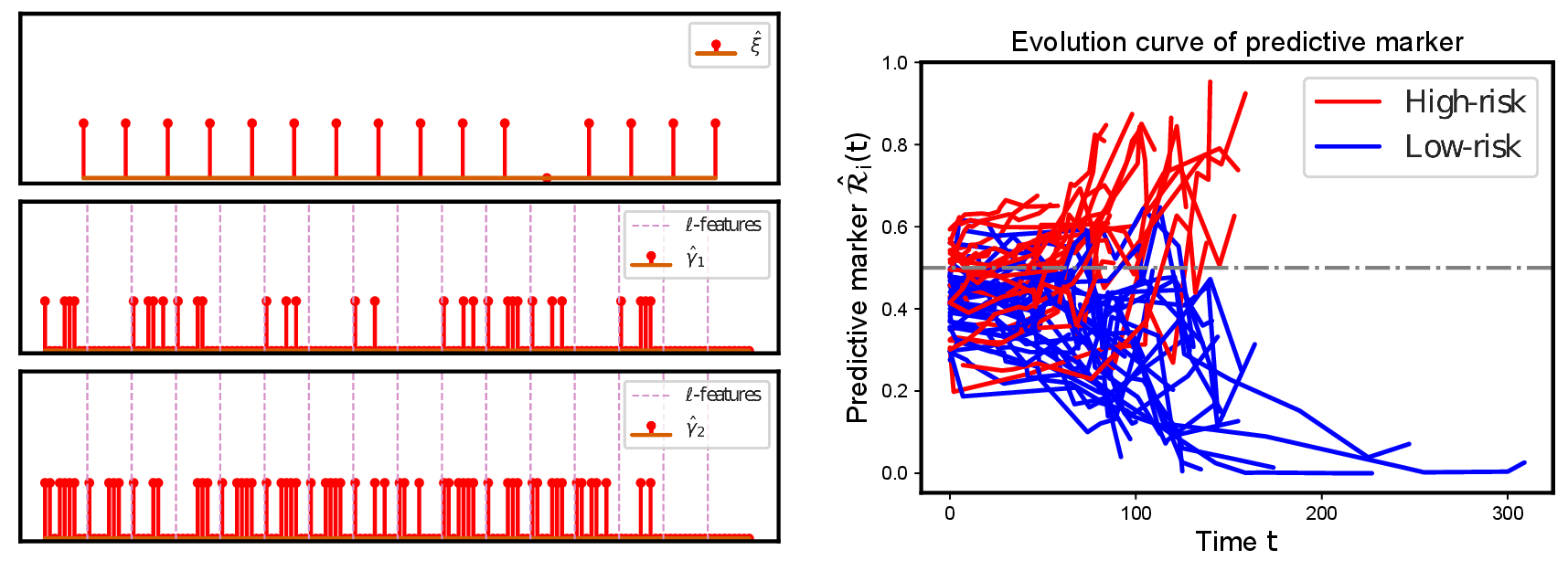}
    \caption{NASA dataset results. Left: in red the support of the estimated coefficient $\hat{\xi}$ ab $\hat{\gamma}_k$ for $k \in \{1, 2 \}$, the dashed pink lines separate the features corresponding to each longitudinal marker $\ell$. Right: the evolution curves of the predictive marker $\widehat \cR_i(t)$ for varying times $t$ and each subject where we well separate the subjects in the high risk group (in red color) and low risk (in blue color) based on their predictive marker at last measurement time with a threshold 0.5 represent by horizontal dashed line.}
    \label{fig:NASA_perf}
\end{figure}

We illustrate in Figure~\ref{fig:NASA_perf} the results obtained by FLASH. In the left panel, we can see the effect of regularization where the coefficients learned by the model are sparse and some longitudinal markers are entirely discarded. In particular, five longitudinal markers are excluded for the first group $k=1$ but not for the group $k=2$ while the last two markers are never selected. In the right panel, we show the evolution in time of the predictive marker for each subject. We can see that, as time passes, more data is observed and the subjects are better separated into two groups of different risks.

\subsection{Experiments on using signatures as association functions}
\label{sec:sig_transform}
\paragraph{Signature transform}
The association features can be extracted by a signature transformation. We refer to \citet{fermanian2021embedding} for a detailed presentation of this transformation and simply recall here its definition. Let $I = (\ell_1,\dots,\ell_k)$ be a word of size $k$ from the alphabet $\{1,\dots,L\}^k$. The signature associated to $I$ is defined as the mapping
    \begin{align*}
        t \mapsto \mathbf{S}^I(\cY_i(t^{-})) :=  \int_{0 < u_1 < \dots < u_k < t} dy_i^{\ell_1}(u_1) \dots dy_i^{\ell_k}(u_k).
    \end{align*}
The signature of depth $N$ is defined as the vector $$\mathbf{S}_N(\cY_i(t^{-})) = \Big(\mathbf{S}^I (\cY_i(t^{-}))\Big)_{|I| \leq N}.$$ 
It is a transformation from a multivariate longitudinal marker to a sequence of coefficients, that are independent of time parameterization and encodes geometric properties (for example, the second order coefficients correspond to areas). It is therefore a very different transformation from the ones used in the \texttt{tsfresh} package, in particular because it encodes information on interactions between coordinates, whereas \texttt{tsfresh} focuses on univariate features. The truncation depth $N$ is an hyperparameter that can be selected typically by cross-validation.

\paragraph{Results} Figure~\ref{fig:flash_sig} shows the performance of FLASH with signatures as association features compared with the performance of the model with the feature extracted from \textit{tsfresh} and the two competing methods on the four datasets presented in the article. The prediction performance of the model with the signature transformation is comparable to the competing methods and better for the FLASH\_simu dataset, and its computation cost is reduced since it does not require to implement a screening phase procedure to select the top association features.

\begin{figure}[ht]
    \centering
    \includegraphics[scale=.4]{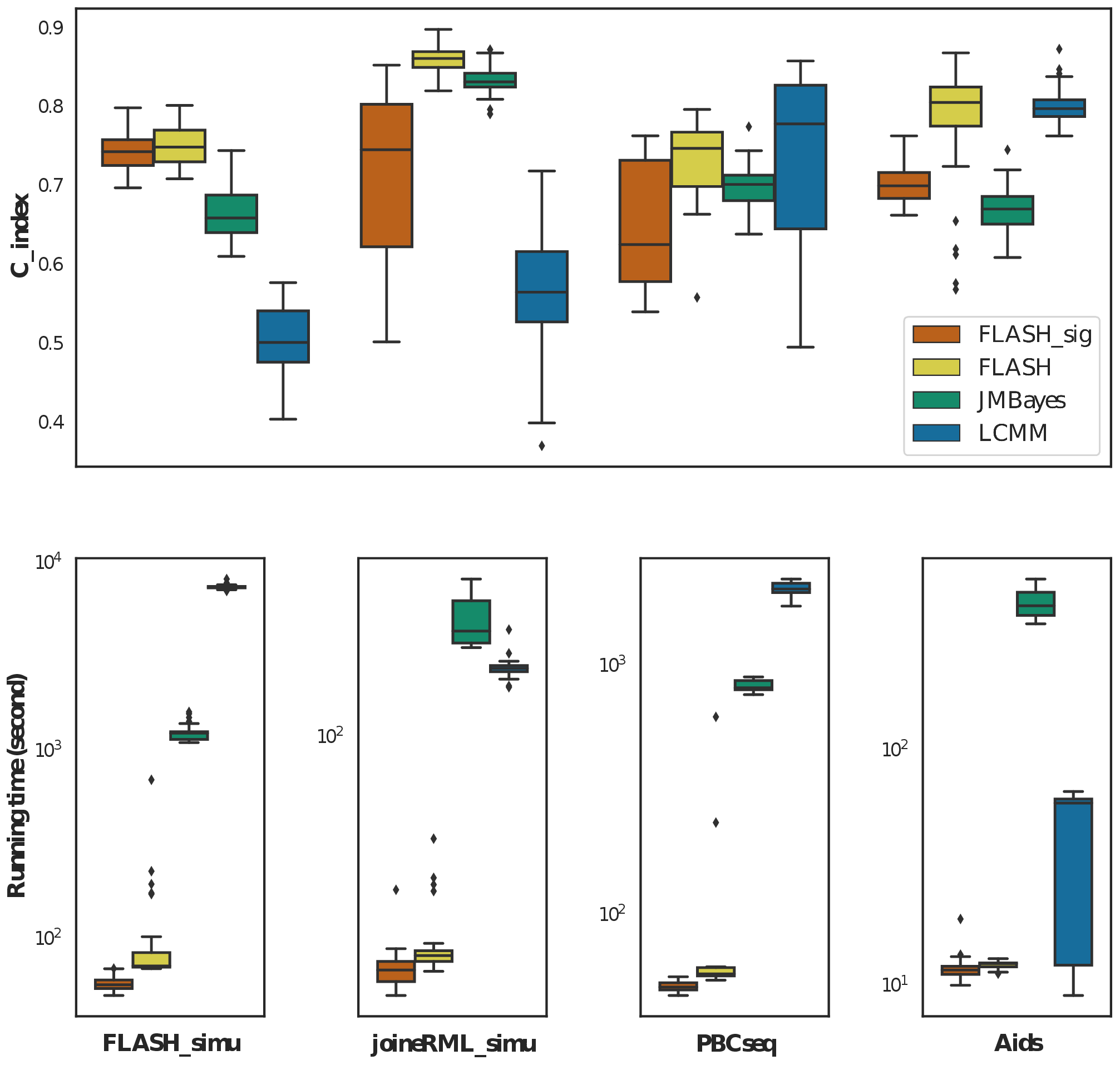}
    \caption{The performance of the Flash model with the signature transform compares with the competing methods.}
    \label{fig:flash_sig}
\end{figure}

\subsection{Procedure to select the optimal number of latent groups}
\label{sec:K_sel}

The optimal number of latent groups, $K$, is selected based on the Bayesian information criterion (BIC)~\citep{hastie2009elements}, which is defined as

\begin{equation*}
    \text{BIC}(K) = -2\hat{\cL}_n^\text{pen}(\theta; K) + \log(n) K,
\end{equation*}
where $\hat{\cL}_n^\text{pen}(\theta; K)$ is the optimal value of the likelihood function (defined in \eqref{eq:pen-log-lik} the main paper) with $K$ groups and $n$ is the number of subjects. The optimal $K$ is selected with the ``elbow method", that is, we pick the value of $K$ corresponding to the first large drop in the BIC values. For example, Figure \ref{fig:nb_K} shows the curve of the BIC values obtained with different $K$ on the PBCseq dataset. In this case, the optimal $K$ is 4.
\begin{figure}[ht]
    \centering
    \includegraphics[scale=.4]{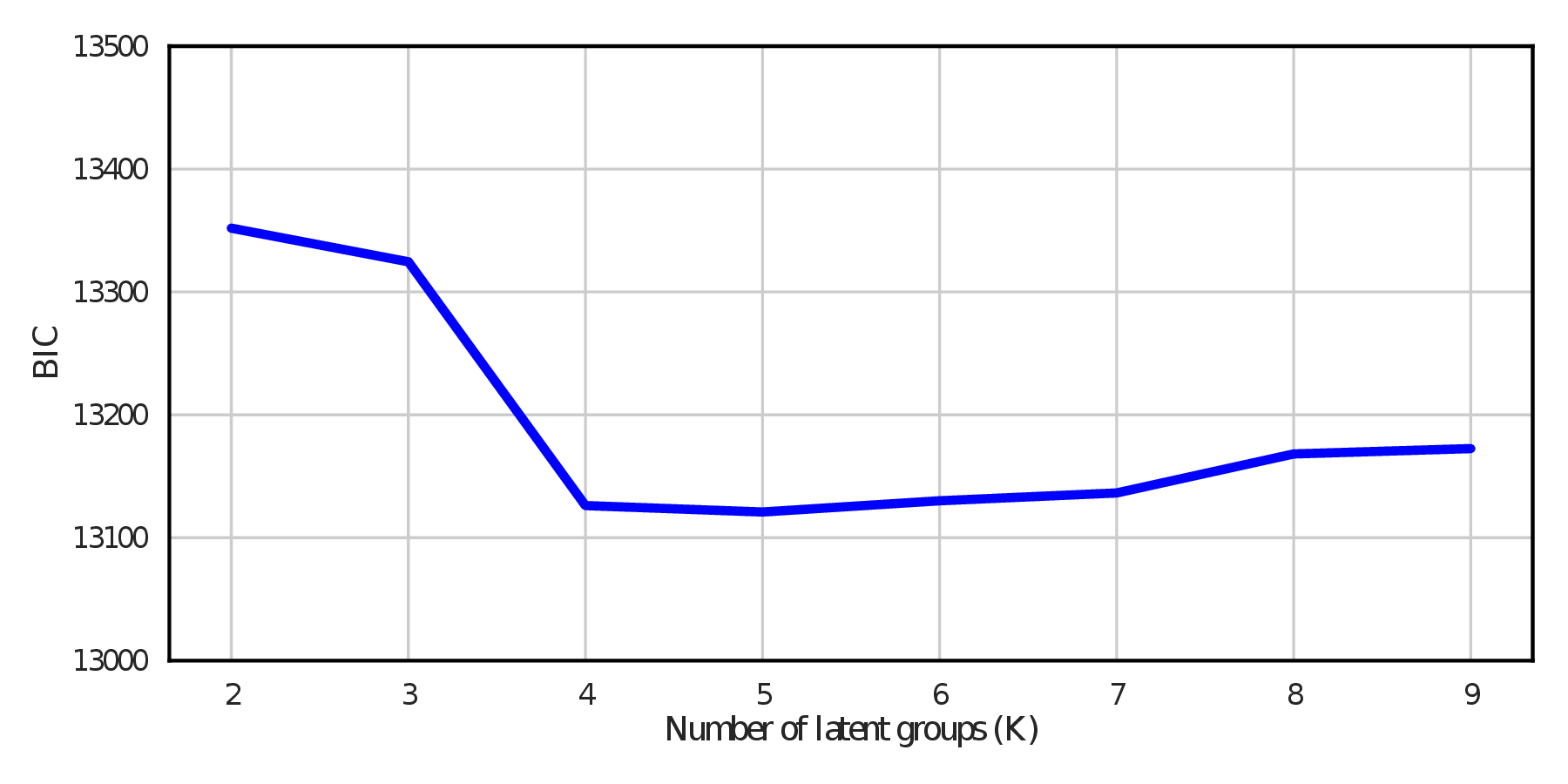}
    \caption{The BIC values on the PBCseq dataset.}
    \label{fig:nb_K}
\end{figure}

\subsection{Sensitivity to latent class assumptions}
\label{sec:sensity}
To assess the sensitivity of the method, we have run an additional simulation where we deviate from the assumption that the latent class is generated from~(\ref{eq:pi}) in the main paper. More precisely, we have used a probit model~\citep{aldrich1984linear}, where the probability of belonging to a class has the form:
  \[ \P(g_i=k)  = \Phi(X_i^\top \xi_k),\]
  where $\Phi$ is the cumulative distribution function of the standard normal distribution. In Figure~\ref{fig:probit} below, we compare the performance of our method on the dataset simulated in the ``well-specified'' setting, described in Section 5.1 of the paper, and on a dataset simulated in this misspecified setting.
   \begin{figure}[ht]
    \centering     \includegraphics[scale=.6]{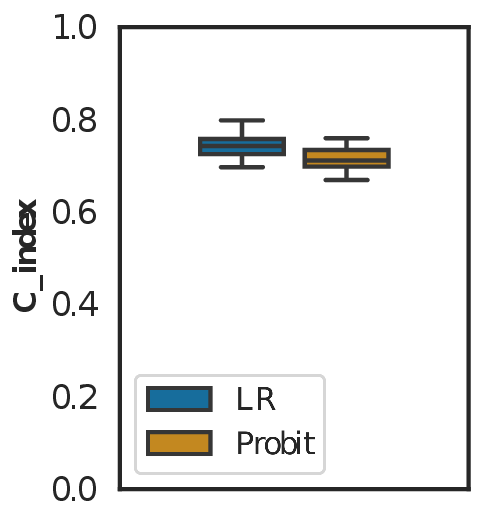}
    \caption{The performance of FLASH in terms of C-index, on two simulated datasets: the LR (Logistic Regression) dataset  is the one simulated with the model~(\ref{eq:pi}) in the main paper and the Probit dataset described above.}
    \label{fig:probit}
\end{figure}
  We can see that the performance is slightly better in the Logistic Regression case, which was to be expected, but that this difference is not significant. Therefore, our proposed method is not too sensitive to the assumption of the model on the probability of latent class membership.

\end{appendices}
\label{lastpage}

\end{document}